\documentclass[english,10pt]{article}
\usepackage{geometry}
\usepackage{float}
\usepackage{mathrsfs,bm}
\usepackage{amsmath}
\usepackage{amssymb}
\usepackage{fancyhdr}
\usepackage{epsfig}
\usepackage{subfig}
\usepackage{amsthm}
\usepackage{mathrsfs}
\usepackage{amsfonts}
\usepackage{mathrsfs}
\usepackage{enumerate}

\makeatletter

\floatstyle{ruled}
\newfloat{algorithm}{tbp}{loa}
\providecommand{\algorithmname}{Algorithm}
\floatname{algorithm}{\protect\algorithmname}

\@ifundefined{definecolor}{\@ifundefined{definecolor}
 {\@ifundefined{definecolor}
 {\usepackage{color}}{}
}{}
}{}

\newtheorem{theorem}{Theorem}[section]
\newtheorem{lem}{Lemma}[section]
\newtheorem{rem}{Remark}[section]

\newcounter{hypA}
\newenvironment{hypA}{\refstepcounter{hypA}\begin{itemize}
  \item[({\bf A\arabic{hypA}})]}{\end{itemize}}

\usepackage{babel}

\begin{document}

%+Title
\begin{center}

{\Large \textbf{Approximate Bayesian Computation for Smoothing}}

\vspace{0.5cm}

BY JAMES S. MARTIN$^{1}$, AJAY JASRA$^{2}$, SUMEETPAL S. SINGH$^{3}$, NICK WHITELEY$^{4}$ \& EMMA McCOY$^{5}$

{\footnotesize $^{1}$Australian School of Business,
University of New South Wales, Sydney, 2052, AUS.}\\
{\footnotesize E-Mail:\,}\texttt{\emph{\footnotesize james.martin04@ic.ac.uk}}\\
{\footnotesize $^{2}$Department of Statistics \& Applied Probability,
National University of Singapore, Singapore, 117546, SG.}\\
{\footnotesize E-Mail:\,}\texttt{\emph{\footnotesize staja@nus.edu.sg}}\\
{\footnotesize $^{3}$Department of Engineering,
University of Cambridge, Cambridge, CB2 1PZ, UK.}\\
{\footnotesize E-Mail:\,}\texttt{\emph{\footnotesize sss40@cam.ac.uk}}\\
{\footnotesize $^{4}$Department of Mathematics,
University of Bristol, Bristol, BS8 1TW, UK.}\\
{\footnotesize E-Mail:\,}\texttt{\emph{\footnotesize nick.whiteley@bristol.ac.uk}}\\
{\footnotesize $^{5}$Department of Mathematics,
Imperial College London, London, SW7 2AZ, UK.}\\
{\footnotesize E-Mail:\,}\texttt{\emph{\footnotesize e.mccoy@ic.ac.uk}}
\end{center}

\begin{abstract}

We consider a method for approximate inference in hidden Markov models (HMMs). The method circumvents the need to evaluate conditional densities of observations given the hidden states. It may be considered an instance of Approximate Bayesian Computation (ABC) and it involves the introduction of auxiliary
variables valued in the same space as the observations. The quality of the approximation may be controlled to arbitrary precision through a parameter $\epsilon>0$. We provide theoretical results which quantify, in terms of $\epsilon$, the ABC error in approximation of expectations of additive functionals with respect to the smoothing distributions. Under regularity assumptions, this error is $\mathcal{O}(n\epsilon)$, where $n$ is the number of time steps over which smoothing is performed. For numerical implementation we adopt the forward-only sequential Monte Carlo (SMC) scheme of \cite{dds2} and quantify the combined error from the ABC and SMC approximations. This forms some of the first quantitative results for ABC methods which jointly treat the ABC  and simulation errors, with a finite number of data and simulated samples.
When the HMM has unknown static parameters, we consider particle Markov chain Monte Carlo \cite{andrieu} (PMCMC) methods for batch statistical inference.\\

\end{abstract}

\section{Introduction}\label{sec:intro}

Hidden Markov Models are widely used in statistics; see \cite{cappe} for a recent overview. An HMM\ is a pair of discrete-time stochastic
processes, $\left\{  X_{n}\right\}  _{n\mathbb{\geq}0}$ and $\left\{
Y_{n}\right\}  _{n\geq 0}$, where $X_{n}\in\mathbb{R}^{d_{x}}$ is unobserved
and $Y_{n}\in\mathbb{R}^{d_{y}}$ is observed. The hidden process
$\left\{  X_{n}\right\}_{n\geq 0} $ is a\ Markov chain with initial density $\eta_0 $ at time $0$ and transition density $f\left(x_{n-1},x_{n}\right)
$, i.e.%
\begin{equation}
\mathbb{P}(X_{0}\in A)=\int_A \eta_0(x)dx \quad\text{ and }%
\quad\mathbb{P}(X_{n}\in A|X_{n-1}=x_{n-1})=\int_{A}f(x_{n-1},x_n
)dx_{n}\quad n\geq1 \label{eq:evol}%
\end{equation}
where $A\subseteq\mathbb{R}^{d_{x}}$ and
$dx_{n}$ is a dominating $\sigma-$finite measure.
Each observation
$Y_{n}$ is conditionally independent of other variables given $\left\{  X_{n}\right\}_{n\geq 0}$ and its conditional distribution is specified by a density
$g\left(x_{n},y_{n}\right)  $, i.e.%
\begin{equation}
\mathbb{P}(Y_{n}\in B|\{X_{k}\}_{k\geq0}=\{x_{k}\}_{k\geq 0})=\int_{B}%
g(x_{n},y_{n})dy_{n}\quad n\geq 0  \label{eq:obs}%
\end{equation}
with $B\subseteq\mathbb{R}^{d_{y}}$ and $dy_{n}$ is a dominating $\sigma-$finite measure. We remark that $\eta_0$, $f(x_{n-1},x_n)$ and $g(x_n,y_n)$ may depend upon
time-independent parameters, which we term static parameters.
%The HMM is given by %equations (\ref{eq:evol})-(\ref{eq:obs}) and is often referred to in the literature as a state-space model.

A variety of inference and estimation tasks for HMMs involve the computation of the smoothing functional $\mathcal{V}_n:\mathbb{R}^{(n+1)d_x}\rightarrow\mathbb{R}$
\begin{equation}
\mathbb{E}[\mathcal{V}_n(X_{0:n})|y_{0:n}] \label{eq:exp_interest}
\end{equation}
where $\mathcal{V}_n(x_{0:n})=\sum_{p=0}^n v_p(x_{p-1:p})$, $v_p:\mathbb{R}^{2d_x}\rightarrow\mathbb{R}$, $1\leq p \leq n$,
$x_{-1:0}=x_0$ with $v_0:\mathbb{R}^{d_x}\rightarrow\mathbb{R}$ and the expectation is w.r.t.~the joint smoothing distribution. For example in the cases $v_p(x_p)=x_p/(n+1)$ and $v_p(x_{p-1:p})=x_{p-1}x_p/(n+1)$,  \eqref{eq:exp_interest} approximates the posterior mean and first-order auto-covariance. When the HMM includes unknown static parameters, expectations of additive functionals play a central role in expectation-maximisation (EM) algorithms and the calculation of score vectors; see \cite{cappe,dds2} for some discussion. 

In practice, the expectation in \eqref{eq:exp_interest} can rarely be computed exactly and one resorts to the use of numerical integration techniques such as SMC; see \cite{doucet} for a recent overview. These methods typically rely on the ability to evaluate pointwise the conditional density $g(x,y)$. The methods we consider address the problem of obtaining an approximation of \eqref{eq:exp_interest} without performing any such evaluations, in a principled manner admitting control of the error in approximation. The motivations for avoiding such evaluations are as follows. Firstly, for some models, $g(x,y)$  may simply not have a closed form expression. Secondly, in some situations evaluation of $g(x,y)$ may be very expensive. %Thirdly, for some models the SMC method can have a computational
%cost of $\mathcal{O}(\kappa^{d_x})$, $\kappa>1$ in order to work well; one can reduce
%this cost by avoiding the evaluation of $g(x,y)$ - see \cite{beskos} and the references therein.

The general technique we consider may be interpreted as an instance of ABC. A recent review of this class of methods can be found in \cite{robert}. In the context of HMMs, ABC has been considered by \cite{dean,jasra}, see also \cite{cornebise}. The approximation (given in Section \ref{sec:abc_approx}) has been introduced by \cite{mckinley} and still requires numerical (e.g.~SMC) methods 
to fit them.
Alternative ideas include nonparametric filtering \cite{gauchi} and the related convolution particle filter \cite{campilo}; see \cite{jasra} for some discussion relative to ABC. As noted by \cite{robert}, in the scenario where there is a fixed amount of data and a fixed number of simulated samples, there is a distinct lack of theoretical results which quantify the combined ABC and numerical (SMC) errors; we provide some of the first results in this context (see \cite{calvet,jasra} for other results). 

Our main objectives are to 
\begin{enumerate}
        \item investigate, both theoretically and empirically, the error associated with the approximation scheme we propose
        \item demonstrate how this scheme can be used to perform smoothing and to estimate static parameters in HMMs from a batch of data
\end{enumerate}

Regarding point 1., the error has two components. The first component arises from the introduction of an auxiliary HMM, incorporating auxiliary variables valued in the same space as the observations $\{y_n\}_{n\geq 0}$. Smoothing expectations under this auxiliary model, which we write as $\mathbb{E}^\epsilon[\mathcal{V}_n(X_{0:n})|y_{0:n}]$, may be taken as approximations of \eqref{eq:exp_interest}; the degree of approximation is controlled through a parameter $\epsilon$ and the error should disappear as $\epsilon\rightarrow 0$. In turn, expectations under this auxiliary model admit efficient numerical approximation using SMC techniques, without evaluation of $g(x,y)$. The second component of the overall error arises from this Monte Carlo scheme, and is controlled through a sample size parameter $N$; the error disappears as $N \rightarrow +\infty$. 
It is noted that the SMC method adopted is the forward only smoothing implementation of the forward-filtering backward smoothing (FFBS) method \cite{andrieu1,godsill} in \cite{dds1}; this is currently one of the most accurate methods for SMC approximation of smoothed additive functionals. 

We will write the SMC estimate of the smoothing expectation under the auxiliary HMM as $\Xi_n[\epsilon,N,\mathcal{V}_n,y_{0:n}]$ and, similarly, we may denote the SMC estimate of the smoothing expectation under the original HMM \eqref{eq:evol}-\eqref{eq:obs} as $\Xi_n[N,\mathcal{V}_n,y_{0:n}]$. In the numerical studies in Section \ref{sec:numerics}, the errors associated with these SMC estimates will be denoted as $e^{N,\epsilon}_n$ and $e^N_n$, respectively.

The overall error associated with the SMC estimate of the smoothing expectation under the auxiliary HMM may be decomposed as  
\begin{align}
\mathbb{E}[\mathcal{V}_n(X_{0:n})|y_{0:n}]-\Xi_n[\epsilon,N,\mathcal{V}_n,y_{0:n}]=&\mathbb{E}[\mathcal{V}_n(X_{0:n})|y_{0:n}]-\mathbb{E}^{\epsilon}[\mathcal{V}_n(X_{0:n})|y_{0:n}]\nonumber\\
+&\mathbb{E}^{\epsilon}[\mathcal{V}_n(X_{0:n})|y_{0:n}]-\Xi_n[\epsilon,N,\mathcal{V}_n,y_{0:n}].\label{intro:err:decomp}
\end{align}
The first difference on the right of \eqref{intro:err:decomp} is a deterministic error, the second difference is a stochastic error. We will provide theoretical analysis of these two error terms which shows how the interplay between $\epsilon$, $n$ and $N$ controls the overall quality of the approximation. These theoretical results, to an extent, are also studied from an empirical perspective.

%***do we need to get into the following details yet?****
%
%The error is decomposed into two parts: The SMC approximation of the ABC approximation and the bias between the ABC approximation and the true smoothing distribution. 
%The error of the SMC approximation of the ABC approximation is shown to grow linearly with the number of data (for some additive functionals), but decrease as $\epsilon>0$ increases. When $N$ is the number of Monte Carlo samples, $n+1$ the number of data and $\delta(\epsilon)$ some decreasing function of $\epsilon$ the error is at most $\mathcal{O}(N^{-1/2}n\delta(\epsilon))$.
%The bias of the ABC approximation is $\mathcal{O}(n\epsilon)$, which grows with $\epsilon$ as expected. The dependence on $n$ changes for empirical
%averages $\mathcal{V}_n(x_{0:n})=\sum_{p=0}^nv_p(x_p)/(n+1)$; see Remark \ref{rem:error_avg}.
%
%*****************************

Regarding point 2., we show how the approximation scheme can be incorporated into a particle MCMC scheme in order to estimate static parameters of the HMM. 
Particle MCMC uses SMC techniques to generate proposals, for example, associated to hidden
states of the HMM.
Here, we use the particle marginal Metropolis-Hastings algorithm in \cite{andrieu} to sample from the ABC approximation of the HMM, with a prior placed upon the unknown static-parameters. In contexts where additive
functions of the hidden state are also of interest (as above), we use the forward-only smoothing
technique mentioned above, to use all the simulated samples from the SMC proposal (which
is not typically adopted). A similar idea has been adopted by \cite{olsson}, except using an additional `backward-pass' in the FFBS algorithm, which is not needed.

In summary, our main contributions are to:

\begin{itemize}
        \item quantify, in terms of $\epsilon$ and  $n$, the error  $\mathbb{E}[\mathcal{V}_n(X_{0:n})|y_{0:n}]-\mathbb{E}^{\epsilon}[\mathcal{V}_n(X_{0:n})|y_{0:n}]$ - henceforth referred to as the \emph{ABC error}
        \item quantify, in terms of $\epsilon$, $N$ and $n$, the  error $\mathbb{E}^{\epsilon}[\mathcal{V}_n(X_{0:n})|y_{0:n}]-\Xi_n[\epsilon,N,\mathcal{V}_n,y_{0:n}]$  - henceforth referred to as the \emph{SMC error} 
        \item provide empirical evidence which illustrates some of these theoretical findings 
        %\item propose and demonstrate SMC methodology which allows numerical approximation %of \eqref{eq:exp_interest}  and permits the estimation of static parameters in the HMM without %the need to evaluate $g(x,y)$
\end{itemize}

This paper is structured as follows. In Section \ref{sec:approach} we discuss the ABC approximation and we characterize the ABC error. In Section \ref{sec:simulation}, SMC and MCMC simulation techniques for targeting the (ABC) smoother are detailed, and the estimation of static parameters in a Bayesian manner is also considered.
In Section \ref{sec:theory} our main theoretical result is given, which combines the ABC and SMC errors discussed above. In Section \ref{sec:numerics} some numerical studies are presented. In Section \ref{sec:summary} the article is concluded. The proofs are in the appendices.

\subsection{Notations}
Given a measurable space $(E,\mathscr{E})$, let $\mu$ be a $\sigma-$finite measure, $K$ be a non-negative kernel and $f:E\rightarrow\mathbb{R}$ a measurable function. The conventions $\mu(f):=\int_{E}f(x)\mu(dx)$, $K(f)(x):=\int_E K(x,dy)f(y)$, $\mu K(f) := \mu(K(f))$ are used. In addition, let  $\textrm{Osc}(f):=\sup_{(x,y) \in E^2}|f(x)-f(y)|$ and let $\mathcal{B}_b(E)$ be
the Banach space of bounded and measurable functions on $E$ endowed with the norm $\|f\|:=\sup_{x\in E}|f(x)|$. For two probability measures $\mu_1,\mu_2$ the total variation distance is $\|\mu_1-\mu_2\|_{TV}=\sup_{A\in\mathscr{E}}|\mu_1(A)-\mu_2(A)|$.
For a Markov kernel $K$ the Dobrushin coefficient is $\beta(K):=\sup_{(x,y)\in E^2}
\|K(x,\cdot)-K(y,\cdot)\|_{TV}$. Given a probability space $(\Omega,\mathscr{F},\mathbb{P})$, we write  $\|\cdot\|_p=\mathbb{E}[|\cdot|^p]^{1/p}$ for the $\mathbb{L}_p$ norm under $\mathbb{P}$. For a set $A\in\mathscr{E}$, $\mathbb{I}_A(x)$ is the indicator function.

\section{ABC Approximation}\label{sec:approach}

\subsection{ABC Smoothing Approximation}\label{sec:abc_approx}
The joint smoothing density is
$$
\hat{\eta}_{n}(x_{0:n}) = \frac{\prod_{i=0}^n g(x_i,y_i) \eta_0(x_0)\prod_{i=1}^n f(x_{i-1},x_i)}{\int_{\mathbb{R}^{(n+1)d_x}}\prod_{i=0}^n g(x_i,y_i) \eta_0(x_0)\prod_{i=1}^n f(x_{i-1},x_i) dx_{0:n}}
$$
where we will suppress the dependence on the data on the l.h.s.
In most scenarios of practical interest, one cannot calculate this density pointwise, or compute expectations w.r.t.~the density.
As a result, a numerical approximation of \eqref{eq:exp_interest}, via advanced computational tools, is required.
This problem is further exacerbated when the density $g(x,y)$ is intractable or very expensive to calculate. That is, one cannot evaluate it pointwise and there is no unbiased estimate available. However, we will assume throughout that one can sample from the associated distribution, for any $x\in\mathbb{R}^{d_x}$. 
It is remarked that this latter condition is not completely necessary for any of the subsequent ideas that will appear (see \cite{ehrlich,yidrlim}), but, it will facilitate a more compact exposition.

To introduce ideas, let us momentarily step away from the setting of HMM's; suppose, one is given observations $y\in\mathcal{D}$
associated to some intractable likelihood $g_{\theta}(y)$ with $\theta\in\Theta$ an unknown parameter. Then, Bayesian inference associated to the posterior $\pi(\theta|y)\propto g_{\theta}(y)\pi(\theta)$ is typically not feasible even using advanced computational tools; see \cite{robert}.
To deal with this issue, ABC draws inference from the following
modified posterior density on $\Theta\mathcal{\times D}$ 
\begin{equation}
\pi_{\epsilon}\left(\theta,u|y\right)=\frac{\pi(\theta)g_{\theta}(u)\mathbb{I}_{A_{\epsilon,y}}(u)}
{\int_{A_{\epsilon,y}\times\Theta}\pi(\theta)g_{\theta}(u)du d\theta}\label{eq:targetABC}
\end{equation}
 with $\epsilon>0$ a tolerance level and $u\in\mathcal{D}$
corresponds to some pseudo-observations. The set $A_{\epsilon,y}$
 is defined as follows
\[
A_{\epsilon,y}=\{z\in\mathcal{D}:\rho(s\left(z\right),s\left(y\right))<\epsilon\}
\]
 where $s:\mathcal{D\rightarrow S}$ represents some summary statistics
and $\rho:\mathcal{S\times S\rightarrow}\mathbb{R}^{+}$ a distance
metric.

As noted by \cite{chopin1,jasra}, given an appropriate structure for the likelihood (such as i.i.d.~data) one can often achieve a more accurate approximation by removing the summary statistics and focusing upon the probabilistic structure of the likelihood.
 Returning now to our setting of the HMM specified in section 1, following \cite{jasra,mckinley} we consider the ABC approximation
of the joint smoothing density, for $\epsilon>0$:
\begin{equation}
\hat{\eta}_{n,\epsilon}(x_{0:n}) = \frac{[\prod_{i=0}^n \int_{\mathbb{R}^{d_y}}\phi(\frac{u-y_i}{\epsilon})g(x_i,u)du] \eta_0(x_0)\prod_{i=1}^n f(x_{i-1},x_i)}{\int_{\mathbb{R}^{(n+1)d_x}}[\prod_{i=0}^n \int_{\mathbb{R}^{d_y}}\phi(\frac{u-y_i}{\epsilon})g(x_i,u)du] \eta_0(x_0)\prod_{i=1}^n f(x_{i-1},x_i) dx_{0:n}}
\label{eq:abc_smooth}
\end{equation}
where $\phi\big(\frac{u-y}{\epsilon}\big)$ is a potential which can take the form $\mathbb{I}_{\{u:|(u-y)/\epsilon|<1\}}$ or a probability density function. This particular ABC approximation maintains the Markovian structure of the model, which will help to facilitate computational algorithms. In particular, in order to approximate
$
\mathbb{E}^{\epsilon}[\mathcal{V}_n(X_{0:n})|y_{0:n}]
$
where $\mathbb{E}^{\epsilon}[\cdot|y_{0:n}]$ is an expectation w.r.t.~the ABC smoothing distribution, one will still need to resort to numerical methods such as SMC and MCMC.

From the perspective of a theoretical justification, results on the consistency in estimation (from both classical and Bayesian perspectives) of static parameters, associated to the ABC approximation,
as $n$ grows can be found in \cite{dean,dean1}. There is an intrinsic asymptotic bias, but this bias can be removed by using a noisy version of ABC; see \cite{dean,dean1} for further details.
In this article we do not address the idea of noisy ABC.

\subsection{ABC Error}

In order to ascertain the potential of using an ABC approximation of HMMs, when considering smoothed additive functionals, we investigate the ABC error. The analysis that follows will concentrate on the scenario in which the ABC kernel is \emph{not} the indicator function. In particular, this allows the inclusion of (A\ref{hyp:A}) (\ref{hyp:smooth_abc}), below, which will facilitate the theoretical analysis of the SMC error of the ABC smoother, in Section \ref{sec:theory}. We remark that the below result for the ABC bias can also be established when $\phi\big(\frac{u-y}{\epsilon}\big)$ is an indicator function with some minor modifications to the proof.

In the subsequent analysis we adopt the following assumption.
\begin{hypA} 
\label{hyp:A}
\renewcommand{\labelitemii}{}
\begin{enumerate}
\item{\label{hyp:density_control}There exists a $1<\rho<\infty$ such that for each $x,x'\in\mathbb{R}^{d_x}$ and $y\in\mathbb{R}^{d_y}$
\begin{eqnarray*}
\rho^{-1} & \leq\quad g(x,y) & \leq\quad \rho\\
\rho^{-1} & \leq\quad f(x,x') & \leq\quad \rho.
\end{eqnarray*}
}
\item{There exists a $L<\infty$ such that for every $y,y'\in\mathbb{R}^{d_y}$
$$
\sup_{x\in\mathbb{R}^{d_x}}|g(x,y)-g(x,y')| \leq L |y-y'|
$$
with $|\cdot|$ the $L_1-$norm.
}
\item{\label{hyp:smooth_abc} There exist functions $\overline{\alpha},\underline{\alpha}:\mathbb{R}_+\rightarrow\mathbb{R}_+$
such that for any $y,u\in\mathbb{R}^{d_y}$, $\epsilon\in\mathbb{R}_+$
$$
\underline{\alpha}(\epsilon) \leq \phi\bigg(\frac{y-u}{\epsilon}\bigg) \leq \overline{\alpha}(\epsilon)
$$
with $\delta(\epsilon):=\overline{\alpha}(\epsilon)/\underline{\alpha}(\epsilon)$ monotonically decreasing. In addition $\int_{\mathbb{R}^{d_y}}\phi\big(\frac{y-u}{\epsilon}\big)dy=1$
with $\int_{\mathbb{R}^{d_y}}|y|\phi(y)dy <+\infty$.}
\item{\label{hyp:function} The analysis is associated to additive functionals $V_n:\mathbb{R}^{(n+1)d_x}\rightarrow\mathbb{R}$:
$$
\mathcal{V}_n(x_{0:n}) = \sum_{p=0}^n v_p(x_p).
$$
with $\overline{v} = \sup_{0\leq p <\infty}\|v_p\|<+\infty$.}
\end{enumerate}
\end{hypA}

The assumptions are rather strong and will typically only hold when the observations and hidden states lie on a compact state spaces, they are however, quite typically of assumptions employed in the analysis of SMC and related approximation methods. Of these assumptions, perhaps (A\ref{hyp:A}-\ref{hyp:smooth_abc}) should be discussed. It can be verified when:
$$
 \phi\bigg(\frac{y-u}{\epsilon}\bigg) \propto \exp\bigg\{-\bigg(\frac{y-u}{\epsilon}\bigg)\bigg\}\quad y\in(a,b).
$$
It is remarked that in practice, one selects $ \phi\big(\frac{y-u}{\epsilon}\big)$, so this
is not such a demanding assumption.
%The assumption is used to illustrate the dependence of the SMC error on $\epsilon$,
%which is investigated in Section \ref{sec:theory}.

\subsubsection{Result}

We have the following result, whose proof is in Appendix \ref{app:abc_error}.

\begin{theorem}\label{theo:abc_error}
Assume (A\ref{hyp:A}). Then there exist a $C<+\infty$ such that for any $\epsilon>0$, $n\geq 1$, $y_{0:n}$,
$$
|\mathbb{E}[\mathcal{V}_n(X_{0:n})|y_{0:n}]-\mathbb{E}^{\epsilon}[\mathcal{V}_n(X_{0:n})|y_{0:n}]| \leq C \epsilon (n+1)
$$
where $\mathbb{E}[\cdot|y_{0:n}]$ and $\mathbb{E}^{\epsilon}[\cdot|y_{0:n}]$ are the expectation w.r.t.~the joint smoothing and ABC smoothing distribution.
\end{theorem}

\begin{rem}
The result establishes that the ABC error does not grow any faster than linearly in time or
$\epsilon$. This is important, as it is known that the SMC error when estimating $\mathbb{E}^{\epsilon}[\mathcal{V}_n(X_{0:n})|y_{0:n}]$ also grows at most linearly
in time \cite{dds1}. As a result, the overall error as the time parameter increases will not necessarily be dominated by one source of error (SMC or ABC). This suggests that an ABC approximation can perform reasonably well in general.
\end{rem}

\section{Simulation-Based Methods}\label{sec:simulation}

\subsection{SMC Methods}\label{sec:smc_desc}

In the context of HMMs, SMC algorithms approximate
$\{\hat{\eta}_{n}\}$ recursively by propagating a collection
of properly weighted samples, called particles, using a combination
of importance sampling and resampling steps. For the importance sampling
part of the algorithm at each step $n$ of the algorithm we will use
general proposal (Markov) kernels $H_{n}$
which possess normalizing constants that do not depend on the simulated
paths. A typical SMC algorithm is given below (we assume it terminates at time $p+1$):

\begin{itemize}
\item {0. Initialisation: set $n=0$; for $i\in\{1,\dots,N\}$ sample
$X_{0}^{(i)}\sim\eta_{0}$ and compute 
\[
G_{0}(x_0^{i}) = %\frac{
g(x_0^{i},y_0)%}{\eta_0(x_0^i)}
\]
 with $W_{0}^{i}=G_{0}(x_0^{i})$. } 
\item {1. Decide whether or not to resample, and if this is performed,
set all weights $\{W_{n}^{(i)}\}_{1\leq i\leq N}$ to $1$. Proceed
to step 2.} 
\item {2. Set $n=n+1$, if $n=p+1$ stop, else; for $i\in\{1,\dots,N\}$
sample $X_{n}^{i}|x_{n-1}^{i}\sim H_{n}(x_{n-1}^{i},\cdot)$,
compute \[
G_{n}(x_{n-1:n}^{i}) = \frac{g(x_n^i,y_n)f(x_{n-1}^i,x_n^i)}{H_n(x_{n-1}^i,x_n^i)}
\]
 and set $W_{n}^{i}=G_{n}(x_{n-1:n}^{i})W_{n-1}^{i}$ and return to the
start of step 1.} 
\end{itemize}

\subsection{Some details on resampling}

If one chooses to implement SMC without resampling steps, i.e.~to perform
sequential importance sampling, as time progresses, the variance of
the weights $\{W_{n}^{i}\}_{1\leq i\leq N}$ typically increases.
This has been commonly referred to as the \emph{weight} degeneracy
property. To counter this resampling is used: the particles 
are sampled with replacement, according to the normalized weights
$\{\bar{W}_{n}^{i}\}_{1\leq i\leq N}$ given by $
\bar{W}_{n}^{i}=\frac{W_{n}^{i}}{\sum_{j=1}^{N}W_{n}^{j}}$
 and then each $W_{n}^{i}$ is reset to $1$. We remark that more efficient alternatives are possible; see e.g.~\cite{doucet}.

If one resamples too often, the simulated past of the path of each
particle will be very similar to each other. This has been documented
as the \emph{path} degeneracy problem. A common remedy was to resample
only when an appropriate criterion drops beneath or goes above some threshold.
In the former case, a common criterion is the effective sample size
$\left(\sum_{j=1}^{N}\left(\bar{W}_{n}^{j}\right)^{2}\right)^{-1}$ \cite{Liu2001}.
This approach, however, does not ultimately solve the path degeneracy problem.
Path degeneracy has been a long standing bottleneck when static
parameters $\theta$ are estimated online using SMC methods by augmenting
them with the latent state; see \cite{dds2}. 
Considering the central limit theorem (CLT) associated to the SMC estimate of $\mathbb{E}[\mathcal{V}_n(X_{0:n})|y_{0:n}]$:
$$
\sum_{i=1}^N \bar{W}_n^i \mathcal{V}_n(x_{0:n}^i),
$$
it is remarked that the issue of path degeneracy leads, under very strong conditions on the HMM, to an asymptotic variance in this CLT that grows quadratically in $n$; see \cite{poya}.
% It is remarked that this issue leads, under very
% strong conditions on the HMM, to an asymptotic variance in the central limit theorem (CLT) associated to the SMC
% estimate of $\mathbb{E}[\mathcal{V}_n(X_{0:n})|y_{0:n}]$:
% $$
% \sum_{i=1}^N \bar{W}_n^i \mathcal{V}_n(x_{0:n}^i)
% $$
%  to grow quadratically in $n$; see \cite{poya}.

Suppose one resamples, multinomially, at every iteration, except when
$n=p$. Denote the resampled index of the ancestor of particle $i$
at time $n$ by $a_{n}^{i}\in\{1,\dots,N\}$; this is a random variable
chosen with probability $\bar{W}_{n-1}^{a_{n-1}^{i}}$. Furthermore
the joint density of the sampled particles and the resampled indices
is \begin{equation}
\psi(x_{0:p}^{1:N},\bar{\mathbf{a}}_{0:p-1})=\bigg(\prod_{i=1}^{N}\eta_{0}(x_{1}^{i})\bigg)\prod_{n=1}^{p}\bigg(\prod_{i=1}^{N}
\bar{W}_{n-1}^{a_{n-1}^{i}}H_{n}(x_{n-1}^{a_{n-1}^{i}},x_{n}^{i})\bigg),\label{eq:smc_algo}\end{equation}
 where the complete genealogy of ancestors is denoted as $\bar{\mathbf{a}}_{n}=(a_{n}^{1},\dots,a_{n}^{N})$
and the randomly simulated values of the state as $x_{n}=(x_{n}^{1},\dots,x_{n}^{N})$.
Together they form the following SMC approximations for $\hat{\eta}_{n}$\[
\hat{\eta}_{n}^{N}(dx_{0:n})=\frac{1}{N}\sum_{j=1}^{N}\delta_{x_{0:n}^{a_{n}^{j}}}(dx_{0:n})\]
 and an approximation of the normalizing constant 
\begin{equation}
\widehat{Z}_{p}=\prod_{n=0}^{p}\bigg\{\frac{1}{N}\sum_{j=1}^{N}G_{n}(x_{n-1:n}^j)\bigg\}.\label{eq:normal_est_smc}
\end{equation}
 The complete ancestral genealogy at each time can always traced back
by defining an ancestry sequence $b_{0:n}^{i}$ for every $i\in\{1,\dots,N\}$
and $n\in\{0,\dots,p-1\}$, whose elements are given by the backward
recursion $b_{n}^{i}=a_{n}^{b_{n+1}^{i}}$ where $b_{p}^{i}=i$. 
%In this context one can view SMC approximations as random probability
%measures induced by the imputed random genealogy $\bar{\mathbf{a}}_{n}$
%and the state sequence $\bar{u}_{n}$. 
This interpretation of SMC
approximations was introduced in \cite{andrieu} and will be used later 
together with $\psi(x_{0:p}^{1:N},\bar{\mathbf{a}}_{0:p-1})$ for
describing PMCMC.

\subsection{Forward only Smoothing}\label{sec:fos}

Due to the path degeneracy effect, one does not want to use the SMC approximation $\hat{\eta}_{n}^{N}(dx_{0:n})$ to perform smoothing.
One potential solution to this issue is the forward filtering backward smoothing algorithm and in particular the forward only implementation of 
it in \cite{dds2}. That is, the FFBS algorithm includes a backward simulation step, which is eliminated in \cite{dds2}.
We consider the SMC approximation of the expectation $\mathbb{E}[\mathcal{V}_n(X_{0:n})|y_{0:n}]$ where
$\mathcal{V}_{n}(x_{0:n})=\sum_{p=0}^n v_p(x_{p-1:p})$.

The construction of the procedure is as follows. It is first noted that
$$
\mathbb{E}[\mathcal{V}_n(X_{0:n})|y_{0:n}] = \int V_n(x_n) \hat{\eta}_n(x_{0:n}) dx_{0:n}
$$
with, for $n\geq 1$
$$
V_n(x_n) := \int \mathcal{V}_n(x_{0:n}) \hat{\eta}_n(x_{0:n-1}|x_n) dx_{0:n-1}
$$
$V_0(x_0)=0$, where $\hat{\eta}_n(x_{0:n-1}|x_n) = \hat{\eta}_n(x_{0:n})/\int \hat{\eta}_n(x_{0:n}) dx_{0:n-1}$,
then one can establish, e.g.~\cite{dds2}, that
$$
V_n(x_n) = \int [V_{n-1}(x_{n-1}) + v_{n}(x_{n-1:n})] \hat{\eta}_n(x_{n-1}|x_n)dx_{n-1}
$$
where $ \hat{\eta}_n(x_{n-1}|x_n) = \int \hat{\eta}_n(x_{0:n-1}|x_n)dx_{0:n-2}$.

These recursions lead to the following idea. Given the current particle approximation of the marginal of $\widehat{\eta}_{n-1}$, 
$\{\bar{W}_{n-1}^i,x_{n-1}^{i}\}_{1\leq i \leq N}$ and of $\{V_{n-1}(x_{n-1}^i)\}_{1\leq i \leq N}$ (write this $\{V^N_{n-1}(x_{n-1}^i)\}_{1\leq i \leq N}$), one performs the following:
update the SMC approximation as in Section \ref{sec:smc_desc} and set
\begin{equation}
V^N_n(x_n^i) =  \frac{\sum_{j=1}^N \bar{W}_{n-1}^j f(x_{n-1}^j,x_n^i) [V^N_{n-1}(x_{n-1}^j) + v_n(x_{n-1}^j,x_n^i)] }{\sum_{j=1}^N \bar{W}_{n-1}^j f(x_{n-1}^j,x_n^i)} \quad i\in\{1,\dots,N\}\label{eq:v_update}
\end{equation}
with $V^N_0=v_0$.
Then the SMC approximation of  $\mathbb{E}[\mathcal{V}_n(X_{0:n})|y_{0:n}]$, is exactly:
\begin{equation}
\sum_{i=1}^N \bar{W}_n^i V^N_n(x_n^i) \label{eq:smc_fos_est}.
\end{equation}
It is apparent that the computational cost of this recursion is $\mathcal{O}(N^2)$ per-time step. For functions such as $\mathcal{V}_{n}(x_{0:n})=\sum_{p=0}^n v_p(x_{p-1:p})$, it has
been seen that, under some assumptions, the asymptotic variance in the CLT 
associated to \eqref{eq:smc_fos_est} grows at most linearly in $n$. This is in contrast to growing quadratically at least quadratically in $n$, under similar assumptions, for the standard SMC estimate; see \cite{dds1} and also \cite{douc} for additional theoretical analysis.

\subsection{SMC for ABC}\label{sec:smc_abc}

If one cannot or does not want to compute $g(x,y)$, then the algorithm described in Section \ref{sec:smc_desc} can seldom be implemented. In such contexts, we can easily use an SMC algorithm to approximate the $\hat{\eta}_{n,\epsilon}$ in \eqref{eq:abc_smooth}. 
For example, in \cite{jasra}, at time 0, one samples the signal from $\eta_0$ and the pseudo observations from the likelihood to yield an incremental weight
$$
G_{0}(u_0^{i}) = \phi\bigg(\frac{u_0^{(i)}-y_0}{\epsilon}\bigg)
$$
where $u_0^{i}$ is the pseudo observation at time 0. At subsequent time-points one can sample from the signal transition and likelihood to obtain 
$G_{n}(u_n^{i}) = \phi\big(\frac{u_n^{(i)}-y_n}{\epsilon}\big)$. The selection of $\epsilon$ can be adaptive and different proposals (other than the state-dynamics) can be adopted; we refer to \cite{jasra} for some discussion.

It is remarked that a drawback of the algorithm is that when $d_y$ grows with $\epsilon,N$ fixed, one cannot expect the algorithm to
work well for every $\epsilon$; typically one must increase $\epsilon$ to yield reasonable algorithmic results and this is at the cost of increasing the bias (see Theorem \ref{theo:abc_error}). To maintain $\epsilon$ at a reasonable level, one
must consider more advanced strategies which are not investigated here.

In scenarios where $\phi\big(\frac{u-y}{\epsilon}\big)=\mathbb{I}_{\{u:|(u-y)/\epsilon|<1\}}\big(\frac{u-y}{\epsilon}\big)$,
a potentially better procedure is to use the rejection kernel in \cite{delmoral} (note this differs from the ideas of \cite{cornebise}). In this case, one initializes the SMC algorithm as above. However, at subsequent time-points, $n\geq 1$, one uses the kernel
\begin{eqnarray}
K_{n}((u_{n-1}^{1:N},x_{n-1}^{1:N}),(u_{n}^{i},x_{n}^{i})) & = &
G_{n-1}(u_{n-1}^{i}) H_n((u_{n-1}^{i},x_{n-1}^{i}),(u_{n}^{i},x_{n}^{i})) \nonumber \\ & &+ [1-G_{n-1}(u_{n-1}^{i})]\sum_{j=1}^N \frac{G_{n-1}(u_{n-1}^{j})}
{\sum_{l=1}^N G_{n-1}(u_{n-1}^{l})} H_n((u_{n-1}^{j},x_{n-1}^{j}),(u_{n}^{i},x_{n}^{i})).
\label{eq:rej_kernel}
\end{eqnarray}
In this case, at any given time-step, we will only resample those particles which have $|u_{n-1}-y_{n-1}|>\epsilon$. It has been shown by \cite[pp.~304-305]{delmoral} that this kernel produces a lower asymptotic variance in the CLT than an algorithm which resamples at every time step. It will be of interest to see if this advantage is realized when $N$ is finite, especially versus the dynamic resampling that is mentioned in Section \ref{sec:smc_desc}. This particular
SMC approach is termed `rejection SMC' (RSMC) throughout the article.

When using SMC for the ABC approximation of $\mathbb{E}[\mathcal{V}_n(X_{0:n})|y_{0:n}]$, the procedure in Section \ref{sec:fos} can be followed with only modifications in notations and state-spaces.

\subsection{PMCMC}\label{sec:pmcmc}

In this section we consider the scenario where one has unknown static parameters $\theta\in\mathbb{R}^{d_\theta}$ associated to the HMM. We concentrate upon batch inference. 

Particle Markov Chain Monte Carlo methods are MCMC algorithms operating on an extended state-space and targeting an extended distribution over the random variables appearing in the SMC algorithm. As in standard MCMC the idea is to run an ergodic
Markov chain to obtain samples from the distribution of interest.
The difference lies in the fact that, due to using an SMC approximation to generate a proposal, the invariant distribution of the simulated chain is defined on an extended state space, with an appropriate marginal being the distribution that we are interested in sampling from in the first place.

We will present the 
particle marginal Metropolis-Hastings (PMMH) algorithm of 
\cite{andrieu}. The PMMH algorithm can
sample from the target distribution 
\begin{equation}
\hat{\eta}_{n}(x_{0:n},\theta) \propto \bigg[\prod_{i=0}^n g_{\theta}(x_i,y_i) \eta_0(x_0)\prod_{i=1}^n f_{\theta}(x_{i-1},x_i)\bigg] \pi(\theta)
\label{eq:joint_post}
\end{equation}
where $\pi(\theta)$ is the prior on $\theta$. We concentrate upon the presentation in the scenario that one is interested in the original HMM; the ABC extension is simple and just uses the SMC procedures
described in Section \ref{sec:smc_abc} instead of those at the start of Section \ref{sec:smc_desc}. Note also, that the algorithm is given when using an SMC algorithm that resamples at each time-step; a dynamic
resampling schedule can also be used.

The PMMH algorithm proceeds as follows:
\begin{itemize}
\item {0. Set $\theta(0)$. Sample $x_{0:p}(0)^{1:N},\bar{\mathbf{a}}_{0:p-1}(0)$
from \eqref{eq:smc_algo} (which now depends upon $\theta$). Sample $k\in\{1,\dots,N\}$ from $\bar{W}_{p}^{k}$
and and compute $\widehat{Z}_p(0)$ as in \eqref{eq:normal_est_smc}.
Store $\widehat{Z}_p(0),k(0),x_{0:p}(0)^{1:N},\bar{\mathbf{a}}_{1:p-1}(0),\theta(0)$. Set $i=1$} 
\item {1. Propose a new $\theta'$ from a candidate $q(\theta(i-1),\cdot)$ and $x_{0:p}^{'1:N},\bar{\mathbf{a}}'_{1},\dots,\bar{\mathbf{a}}'_{p-1}$
and $k'$ as in step 0. Accept or reject this as the new state of
the chain with probability \[
1\wedge\frac{\widehat{Z}'}{\widehat{Z}(i-1)} \frac{\pi(\theta')q(\theta',\theta(i-1))}{\pi(\theta(i-1))q(\theta(i-1),\theta')}
\]
 If we accept, set $\left(\widehat{Z}(i),k(i),x_{0:p}^{1:N}(i),\bar{\mathbf{a}}_{1:p-1}(i),\theta(i)\right)=\left(\widehat{Z}',k',x_{0:p}^{'1:N},\bar{\mathbf{a}}'_{1:p-1},\theta'\right)$, otherwise\\
$\left(\widehat{Z}(i),k(i),x_{0:p}^{1:N}(i),\bar{\mathbf{a}}_{1:p-1}(i),\theta(i)\right)$ $=\left(\widehat{Z}(i-1),k(i-1),x_{0:p}^{1:N}(i-1),\bar{\mathbf{a}}_{1:p-1}(i-1),\theta(i-1)\right)$.
Set $i=i+1$ and return to 1. } 
\end{itemize}
In \cite{andrieu} it is shown that the sequence $\{x_{0:p}^{b_{0:p}^{k(i)}}(i),\theta(i)\}_{i\geq 1}$ provides an approximation of \eqref{eq:joint_post}, for any $N\geq 1$.

If one is interested in approximating, say 
$$
\int \mathcal{V}_n(x_{0:n},\theta)\hat{\eta}_{n}(x_{0:n},\theta)d(x_{0:n},\theta)
$$
as noted by \cite{olsson}, the FFBS estimate can be used, based upon the SMC at each time-step, by simply extending the definition of $V_n^N$ in \eqref{eq:v_update}
to include $\theta$
 (c.f.~\eqref{eq:smc_fos_est}):
$$
\frac{1}{M}\sum_{i=1}^M \bigg\{
\sum_{j=1}^N \bar{W}_n^j(i) V^N_n(x_n^j(i),\theta(i))
\bigg\}
$$
where the first summation is over $M$ iterations of the PMMH algorithm. Trivially, one can extend this to the case where only a forward pass as in Section \ref{sec:fos} is used. A critical point is despite the improvement in the SMC estimation, whether this is necessarily
reasonable given the increase in computational cost and the iterative nature of the MCMC; especially in an ABC context, which is presently not known to our knowledge.
We remark again, that any of the SMC for ABC algorithms mentioned in Section \ref{sec:smc_abc} can be adopted, when considering the ABC approximation of \eqref{eq:joint_post}; whether using dynamic resampling or the kernel
\eqref{eq:rej_kernel} the estimate of the normalizing constant is unbiased - see \cite{andrieu} for why this is of interest.

%the invariant density of the Markov
%kernel above is exactly 
%\[
%\overline{\pi}_{p}^{N}(k,\bar{u}_{1:p},\bar{\mathbf{a}}_{1:p-1})=\frac{\overline{\pi}_{p}(u_{1:p}^{(k)})}{N^{p}}\frac{\psi(\bar{u}_{1:p},\bar{\mathbf{a}}_{1:p-1})}{\overline{M}_{1}(u_{1}^{(b_{1}^{k})})\prod_{n=2}^{p}\big\{\bar{W}_{n-1}^{(b_{n-1}^{k})}\overline{M}_{n}(u_{n}^{(b_{n}^{k})}|u_{n-1}^{(b_{n-1}^{k})},\dots,u_{1}^{(b_{1}^{k})})\big\}}\]
% where $\psi$ is as in \eqref{eq:stopping_density} and as before
%we have $b_{p}^{k}=k$ and $b_{n}^{k}=a_{n}^{b_{n+1}^{k}}$ for every
%$k\in\mathbb{T}_{N}$ and $n\in\mathbb{T}_{p-1}$. The target density
%of interest, $\overline{\pi}_{p}$, is the marginal, conditional on
%$k$ and $\bar{\mathbf{a}}_{1:p-1}$.
%

\section{Theoretical Analysis}\label{sec:theory}

\subsection{Set-Up}

We consider the error in estimation of smoothed additive functionals, when using an ABC approximation of the HMM. 
This is in the scenario where one does not need to estimate static parameters. 
Recall that we have already considered the ABC error in Theorem \ref{theo:abc_error}; the main objective
is to present a result with regards to the SMC error and the overall effect on the approximation of $\mathbb{E}[\mathcal{V}_n(X_{0:n})|y_{0:n}]$.
We will use (A\ref{hyp:A})
which only applies in the scenario where one uses a kernel density in the ABC approximation (i.e.~not an indicator function). In addition, the SMC algorithm samples from the transition density of the state, with multinomial resampling at every time step (that is, RSMC is not considered). 
These hypotheses can be removed with a more technical proof.
In addition, we condition upon the data and do not treat the randomness of these quantities. We simply assume that we are given a data set and do not address the issue of whether they may, or may not originate from a HMM.

\subsection{Result}

Below the $\mathbb{L}_p-$norm is associated to the random process generated by the SMC algorithm. We also use the abuse of notation $G_{n,\epsilon}(x) = \phi\big(\frac{x-y_n}{\epsilon}\big)$, $x\in\mathbb{R}^{d_y}$, to represent the incremental weights of the SMC algorithm (that is as described in Section \ref{sec:smc_abc}). Note that, in comparison to \eqref{eq:smc_fos_est}, we resample at every time-point, so we can use the incremental weights in the estimate, instead of the normalized weights. Note that
\begin{equation}
\Xi_n[\epsilon,N,\mathcal{V}_n,y_{0:n}] = \sum_{i=1}^N\frac{G_{n,\epsilon}(x_n^i)}{\sum_{j=1}^N G_{n,\epsilon}(x_n^j)} V_{n,\epsilon}^N(x_n^i)
\label{eq:smc_est}
\end{equation}
where $\Xi_n[\epsilon,N,\mathcal{V}_n,y_{0:n}]$ is the quantity that we discussed in Section \ref{sec:intro}; from herein we use the R.H.S.~of \eqref{eq:smc_est} to denote the SMC estimate.

\begin{theorem}\label{theo:main_theorem}
Assume (A\ref{hyp:A}). There exist a $C<+\infty$ and
for any $p\geq 1$ there exist a $a_p<\infty$ such that for any
$\epsilon>0$, $N\geq 1$, $n\geq 1$ and $y_{0:n}$:
$$
\bigg\|\sum_{i=1}^N\frac{G_{n,\epsilon}(x_n^i)}{\sum_{j=1}^N G_{n,\epsilon}(x_n^j)} V_{n,\epsilon}^N(x_n^i) - \mathbb{E}[\mathcal{V}_n(X_{0:n})|y_{0:n}]\bigg\|_p \leq (n+1)\bigg[\frac{a_p\overline{\delta}(\epsilon)}{\sqrt{N}} + C\epsilon\bigg]
$$
where $\overline{\delta}(\epsilon)=\max\{\delta(\epsilon)\max\{\delta(\epsilon)^2,\delta(\epsilon)^4\}, \delta(\epsilon)^{3}\}$, with $\delta(\epsilon)$ as in (A\ref{hyp:A}) (\ref{hyp:smooth_abc}) and
$\mathbb{E}[\cdot|y_{0:n}]$ is the expectation w.r.t.~the joint smoothing distribution.
\end{theorem}

\begin{proof}
After adding and subtracting $\mathbb{E}^{\epsilon}[\mathcal{V}_n(X_{0:n})|y_{0:n}]$
one can apply Minkowski followed by Theorem \ref{theo:smc_error} and Theorem \ref{theo:abc_error} to conclude. 
\end{proof}

\begin{rem}\label{rem:error_avg}
The bound is decomposed into two sources of error. For the SMC approximation, the error tends to decrease as $\epsilon$ grows as one would expect. Conversely, the ABC error term
$|\mathbb{E}^{\epsilon}[\mathcal{V}_n(x_{0:n})|y_{0:n}]-\mathbb{E}[\mathcal{V}_n(x_{0:n})|y_{0:n}]|$ %(where $\mathbb{E}^{\epsilon}[\cdot|y_{0:n}]$ is the expectation w.r.t.~the joint %ABC smoothing distribution)
 grows as $\epsilon$ grows. Both error rates increase at most linearly with the time parameter.
When $\mathcal{V}_n(x_{0:n}) = \sum_{p=0}^n v_p(x_p)/(n+1)$, one can remove the linear decay in the bias term. In addition, with $N$ fixed, the SMC error can be shown to \emph{decrease} with $n$; see \cite{dds1,dds2}.
\end{rem}

\section{Simulations}\label{sec:numerics}

\subsection{Model}

Our numerical studies are implemented on the following HMM also considered in, for example, \cite{andrieu}. We take $d_x=d_y=d$ and the model is: 
\begin{eqnarray*}
%\label{eq:NLM_trans}
X_n & = & \frac{X_{n-1}}{2} + \frac{25X_{n-1}}{1+X^2_{n-1}} + 8\cos(1.2n)  + \zeta_{X,n}, \;\;\; n\ge1,\\
%\label{eq:NLM_obs}
Y_n & = & \frac{X^2_n}{20} + \zeta_{Y,n} \;\;\; n\ge 0,
\end{eqnarray*}
with
$\zeta_{X,n}\stackrel{\textrm{i.i.d.}}{\sim}  \mathcal{N}_d\left(0,\sigma_X^2 \mathbf{I}_d\right)$ 
and independently
$\zeta_{Y,n} \stackrel{\textrm{i.i.d.}}{\sim} \mathcal{N}_d\left(0,\sigma_Y^2 \mathbf{I}_d\right)$ and
$X_0  = 0_d$
with $0_d$ the $d$-dimensional zero vector, $\mathbf{I}_d$ the $d\times d$ identity matrix and $\mathcal{N}_d(\mu,\Sigma)$
the $d$-dimensional normal distribution of mean $\mu$ and covariance matrix $\Sigma$.
Whilst the conditional density of the observations given the state is not intractable, it will facilitate an investigation into the accuracy of ABC. In this scenario one can obtain an approximation of the `correct' answers using SMC/PMCMC with many particles/iterations. 

The objective of our numerical study, for smoothing is to consider the accuracy of ABC, when only considering forward only smoothing (the performance of forward only smoothing relative to using the path of particles has been studied elsewhere - for example \cite{dds1}). 
We also want to investigate the worth of RSMC in the ABC context; recall the asymptotic improvements predicted in \cite{delmoral}. Along the way we also consider the issue of the dimension
of the HMM and the utility of using ABC in high-dimensions. Finally, the time dependence of the errors are presented, to allow some investigation into Theorem \ref{theo:main_theorem}.
When considering PMCMC, we are concerned with both the accuracy of ABC for batch static parameter estimation and the worth of including forward only smoothing as a `post-processing' of the MCMC output.

\subsection{Smoothing}

\subsubsection{Implementation Details}

We consider estimating the expected mean state over the observation period $[0,100]$; i.e.~$v_{p}(x_p) = x_p/101$, $p\in\{0,\dots,100\}$.
We set $\sigma_X^2=10$ and $\sigma_Y^2=1$.
The data are simulated from the true model with the given parameter values.
To obtain a true answer with which to understand the accuracy of the methods we investigate, we use the mean estimate obtained over 50 implementations of the forward smoothing procedure, targeting the exact model, with $5000$ particles.

The algorithms for SMC (that is, approximating the exact model) and SMC ABC are run for 10 different values of $N\in\{100,200,\dots,
1000\}$ which are labelled $N_1,\dots,N_{10}$ in the Figures. 
The SMC ABC approach, i.e.~that dynamically resamples, does so when the effective
sample size drops below $N/2$. For the SMC, SMC ABC and RSMC (which targets the ABC approximation) the hidden state dynamics are used as proposals. 
We also run the algorithms for $d\in\{1,2,5,10\}$ which will also allow us to assess the accuracy of SMC for an ABC HMM in `high' dimensions. 
% of the ABC against exact implementations; as noted in Section \ref{sec:intro} this is of  particular interest for SMC. 
To investigate the accuracy of ABC, we compute a true value for $\mathbb{E}[\mathcal{V}_n(X_{0:n})|y_{0:n}]$, as discussed above, and then average
the $\mathbb{L}_1-$error of the estimate $\Xi_n[\epsilon,N,\mathcal{V}_n,y_{0:n}]$, calculated with respect to the computed true value, across its dimension, i.e.~$e^{N,\epsilon}_n=\frac1d\left|\mathbb{E}[\mathcal{V}_{n}(X_{0:n})|y_{0:n}]-\Xi_{n}[\epsilon,N,\mathcal{V}_{n},y_{0:n}]\right|$, with $|\cdot|$ the $\mathbb{L}_1$-distance. An SMC procedure targeting the exact HMM is also run, to provide some benchmark performance; the corresponding error, $e^N_n$, is similarly calculated as the dimension-averaged $\mathbb{L}_1-$error of the SMC estimate $\Xi_n[N,\mathcal{V}_n,y_{0:n}]$ with respect to the same true value as above.
%This is also done for an exact SMC implementation, to provide some benchmark performance.
All results are averaged over 50 independent runs.

For the ABC specification, we set $\phi\big(\frac{u-y}{\epsilon}\big)=\mathbb{I}_{ \{u:|u-y/\epsilon<1|\} }(u)$% with $|\cdot|$ the $\mathbb{L}_1-$distance
; this will allow us to easily understand the impact of the RSMC. 
In the implementations of the two SMC ABC schemes described in Section \ref{sec:smc_abc} $\epsilon$ is set to be the smallest obtainable in a preliminary set of runs. That is, the smallest
$\epsilon$ for which the weights do not become zero at any time-point.

To conclude the numerical study, we consider the time-dependence of the bias. We consider the SMC and ABC using only forward only smoothing as the time parameter increases from 10, 20,$\dots$, 100, $d\in\{1,2,5,10\}$. The SMC algorithm is run with $N=1000$. RSMC is not considered.

\subsubsection{Results}

The exact and ABC forward smoothing errors, $e^N_n$ and $e^{N,\epsilon}_n$ respectively, are presented in Figure \ref{fig:ForwardSmoothingErrors_SMCvsABC}; the mean errors obtained across the 50 runs are displayed, along with their standard errors. 
%For the ABC, as one would expect, in almost all of the plots the accuracy of the ABC method cannot improve with increasing $N$ (as the bias persists), but the variability of the estimates falls - i.e.~the SMC error is being controlled.
Under the ABC HMM, as one would expect, in almost all of the plots the accuracy of the estimate $e^{N,\epsilon}_n$ cannot improve with increasing $N$ (as the bias persists), but the variability of the estimates falls - i.e.~the SMC component of $e^{N,\epsilon}_n$ is being controlled.
In the plots for $d\in\{1,2\}$, the exact implementation outperforms, as one would expect, the ABC approximation in terms of accuracy. This is illustrated by the means and standard errors of the smoothing errors $e^N_n$ being smaller than those of the ABC smoothing errors $e^{N,\epsilon}_n$ for a vast majority of the values of $N$.  
Interestingly, as the dimension increases ($d\in\{5,10\}$), the ABC estimates appear to be \emph{more} accurate than their SMC counterparts (at least for this function). One might explain this as follows.
For SMC in high-dimensions, one often requires $N=\mathcal{O}(\kappa^d)$ ($\kappa>1$) for some stability, but this is not the case for ABC - see \cite{beskos} and the references therein. 
%Thus, at least here, it appears that the SMC error dominates the ABC error and one does not increase $N$ enough to make the SMC accurate enough. 
These (empirical) results suggest that ABC is a viable approximation technique in higher-dimensions, where it can be difficult to find SMC techniques that always work well.

\begin{figure}[!h]\centering
\vspace{-5cm}
\includegraphics[width=14cm,height=20cm,angle=0]{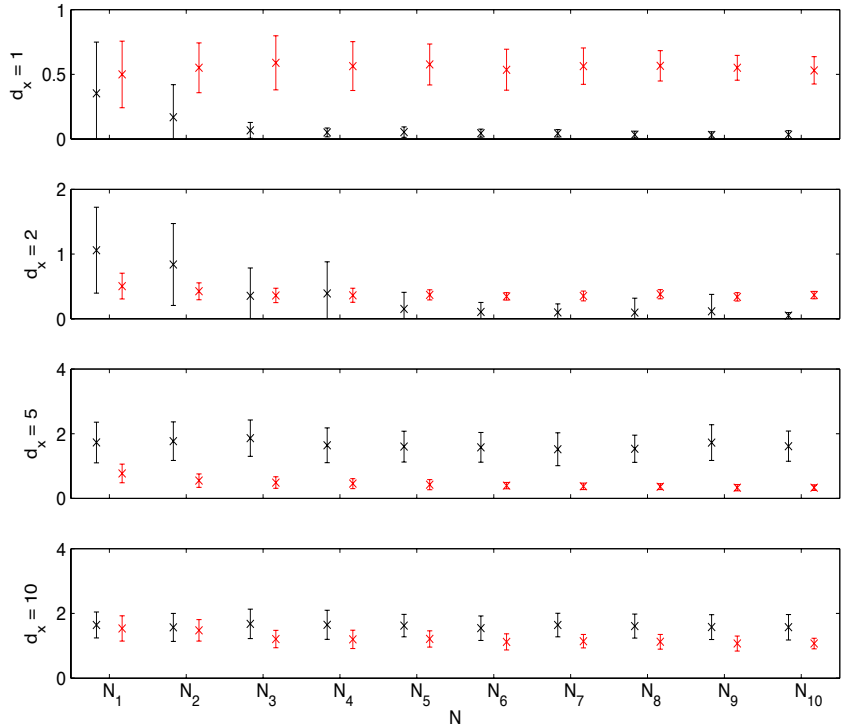}
\vspace{-5cm}
\caption{Mean and standard deviations of the smoothing errors associated with the estimates 
of the mean state ($v_p(x_p)=x_p/101$) obtained over 50 independent implementations of the forward smoothing procedure targeting the true HMM ($e^N_n$, black) and its ABC approximation ($e^{N,\epsilon}_n$, red). The horizontal
axis represents the 10 different values of $N$ for which we ran both algorithms.}
\label{fig:ForwardSmoothingErrors_SMCvsABC}
\end{figure}

\begin{figure}[!h]\centering
\vspace{-5cm}
\includegraphics[width=14cm,height=20cm,angle=0]{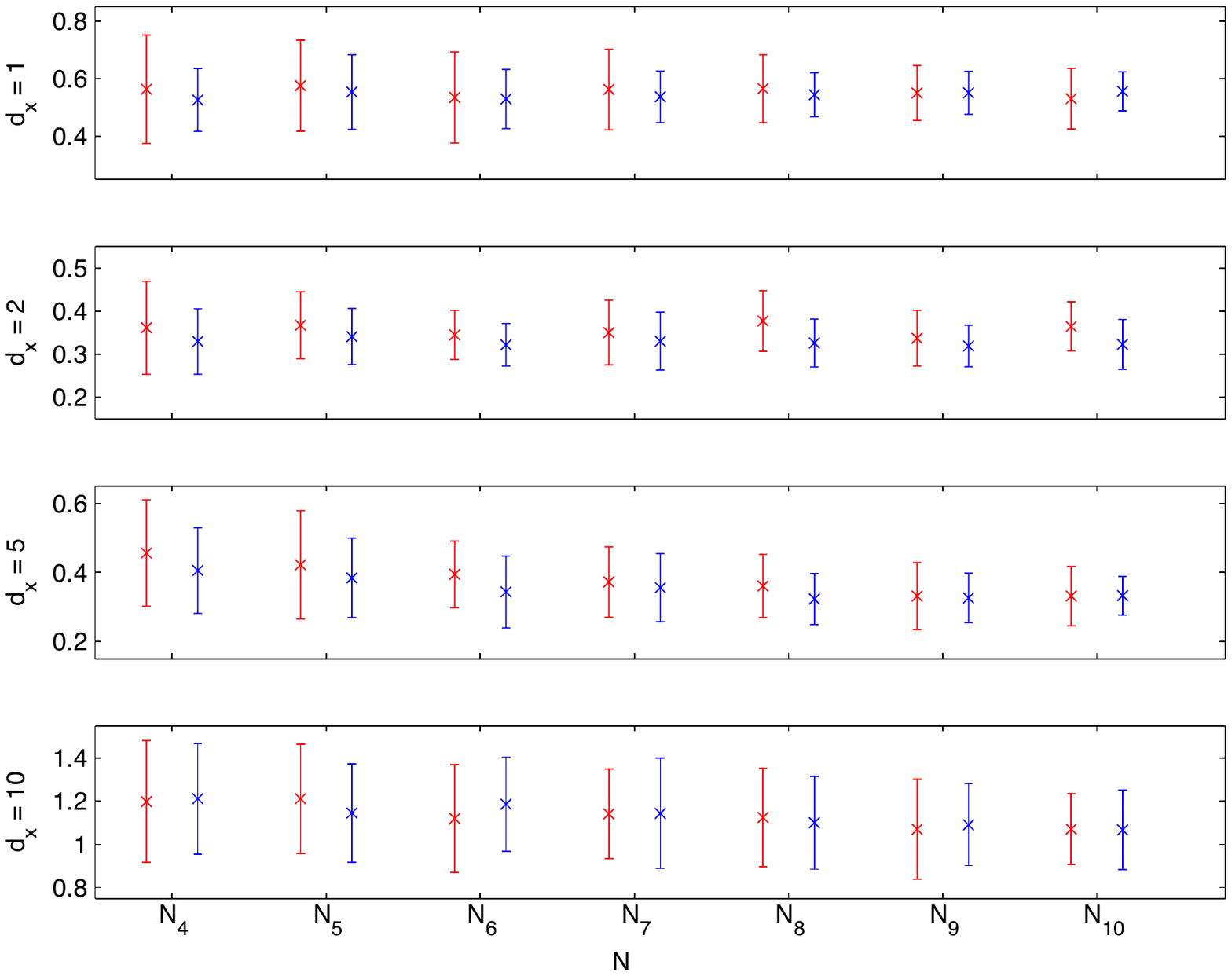}
\vspace{-5cm}
\caption{Mean and standard errors of the smoothing errors associated with the estimates 
of the mean state ($v_p(x_p)=x_p/101$) obtained over 50 independent implementations of the forward smoothing
SMC ($e^{N,\epsilon}_n$, red) and RSMC ($e^{N,\epsilon,R}_n$, blue) procedures targeting the ABC approximation of the smoothing distribution.}
\label{fig:ForwardSmoothingErrors_ABCSMCvsABCRSMC}
\end{figure}

\begin{figure}[!h]\centering
\vspace{-5cm}
\includegraphics[width=14cm,height=20cm,angle=0]{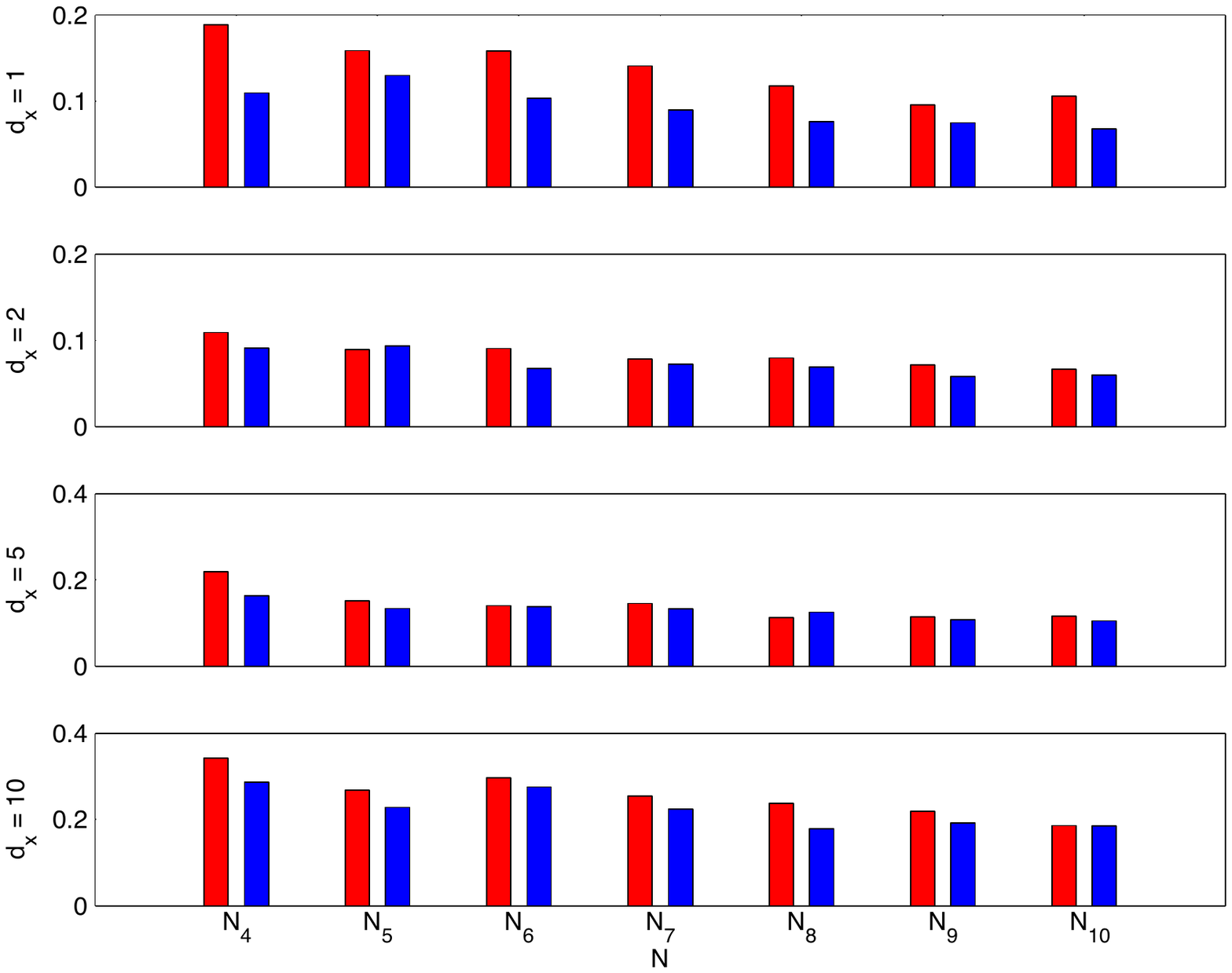}
\vspace{-5cm}
\caption{Standard errors of the smoothing errors associated with the estimates of the mean state ($v_p(x_p)=x_p/101$) obtained over 50 independent implementations of the forward smoothing SMC ($e^{N,\epsilon}_n$, red) and RSMC ($e^{N,\epsilon,R}_n$, blue) procedures targeting the ABC approximation of the smoothing distribution. These standard errors are also displayed in Figure \ref{fig:ForwardSmoothingErrors_ABCSMCvsABCRSMC}.}
\label{fig:ForwardSmoothingStdErrors_ABCSMCvsABCRSMC}
\end{figure}

In Figures \ref{fig:ForwardSmoothingErrors_ABCSMCvsABCRSMC} and \ref{fig:ForwardSmoothingStdErrors_ABCSMCvsABCRSMC} we compare the performance of SMC and RSMC for performing estimation of the smoothing expectation under the ABC HMM. Figure \ref{fig:ForwardSmoothingErrors_ABCSMCvsABCRSMC} presents the mean and standard errors of the ABC smoothing errors that correspond to estimates calculated using the SMC method and the RSMC method; distinction is made through a further subscript, with the ABC smoothing errors corresponding to the RSMC estimates being denoted $e^{N,\epsilon,R}_n$. For clarity of presentation, we only display the results of 7 of the 10 values of $N$ which were run.
From Figure \ref{fig:ForwardSmoothingErrors_ABCSMCvsABCRSMC}, it is noted that the accuracy of the SMC and RSMC procedures for performing ABC forward smoothing are very comparable, with the RSMC estimates even appearing to offer a marginal improvement over the SMC estimates in terms of mean smoothing error. The standard errors of $e^{N,\epsilon}_n$ and $e^{N,\epsilon,R}_n$ are more clearly presented in Figure \ref{fig:ForwardSmoothingStdErrors_ABCSMCvsABCRSMC}. This figure shows that, under the ABC HMM, the variability of the RSMC procedure seems to be slightly less than that of the SMC procedure, especially as $N$ is allowed to grow. In addition, the observed run times for the ABC forward smoothing procedure were consistently lower when using the RSMC method against the SMC ABC approach. These results suggest, at least under the criteria considered, that the use of RSMC would not only be a viable alternative, but it could be preferable to using SMC with dynamic resampling. This is when using forward smoothing to perform inference with respect to the ABC approximation of the HMM.

In Figure \ref{fig:time} we consider the time dependence of the error $e^{N,\epsilon}_n$ associated with the SMC method applied to the ABC HMM. We can observe that in this scenario, there is not any obvious increase in the overall error $e^{N,\epsilon}_n$, with time, for this particular estimate associated to the smoothing distribution. This is consistent with our theoretical results which illustrate that the error does not grow any worse than linearly with time. As expected, on the basis of the results above, the quality of the SMC approximation appears to deteriorate (for $N$ fixed) as the dimension grows, but such a deterioration is less obvious for the ABC approximation.

\begin{figure}[!h]\centering
\vspace{-5cm}
\includegraphics[width=14cm,height=20cm,angle=0]{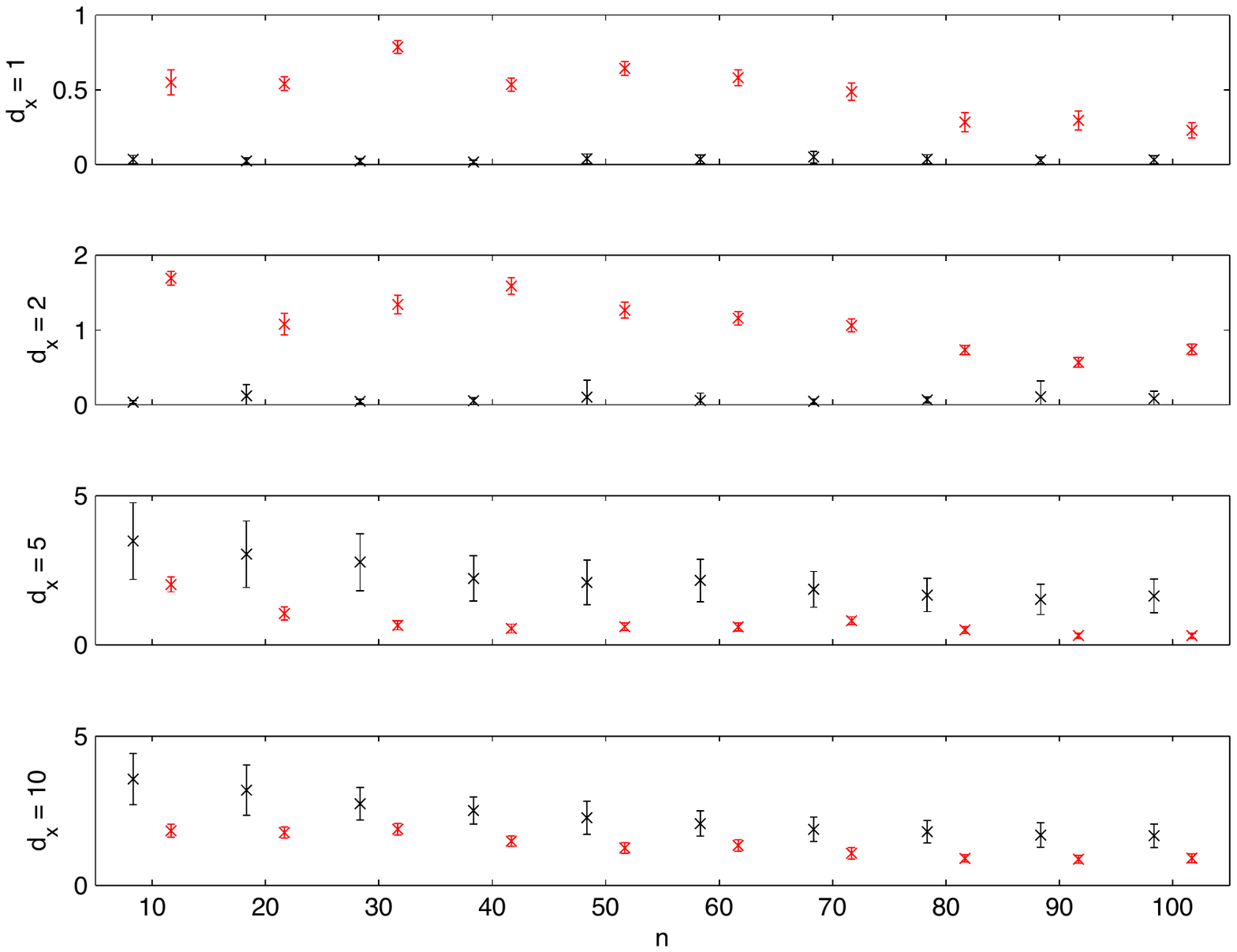}
\vspace{-5cm}
\caption{Time dependence (horizontal axis) of the smoothing errors associated with the estimates 
of the mean state ($v_p(x_p)=x_p/(n+1)$, $n=10,\ldots,100$) obtained over 50 independent implementations of the forward smoothing
SMC ($e^N_n$, black) and ABC ($e^{N,\epsilon}_n$, red) procedures targeting the true and ABC approximation of the smoothing distribution respectively.}
\label{fig:time}
\end{figure}

\subsection{PMMH}

\subsubsection{Implementation Details}

As for the smoothing, we estimate the expected mean state over the observation period $[0,100]$ as well as estimating the static parameters $\theta=(\sigma_X,\sigma_Y)$, with priors
as in \cite{andrieu}. We set $d=1$ throughout and 
the data are the same as for the smoothing experiment (when $d=1$). To obtain a proxy for the true value
we ran a PMMH algorithm as described in \cite{andrieu} for 50000 iterations with 20000 particles (no forward smoothing)
and averaged the results over 50 runs. When no forward smoothing is used, only the selected
particle (see Section \ref{sec:pmcmc}) is used for the estimate of a smoothed additive functional.

The ABC approximation was as for the smoothing example (that is, the function $\phi\big(\frac{u-y}{\epsilon}\big)=\mathbb{I}_{ \{u:|u-y/\epsilon<1|\} }(u)$). To allow direct comparison to running an exact PMMH algorithm (that is, one
which uses a dynamic resampling SMC algorithm on the true HMM) we only adopt an SMC ABC algorithm, i.e. we do not consider the use of RSMC here. For the SMC and SMC ABC the hidden state dynamics are used as proposals. 
The PMMH proposal on the parameters is as in \cite{andrieu}.
We run the algorithms for 50000 iterations with a 10000 iteration burn in. In addition, 5 different
values of $N$ are considered $N\in\{100,200,\dots,500\}$ for the forward only smoothing approaches. In comparing to PMMH algorithms that do not use all the particles (and hence
the computational cost of the SMC algorithm is $\mathcal{O}(N)$) a number of particles
with similar computational costs are run; these were $\{4427,17139,39020,68258,107007\}$.

As with SMC smoothing, the accuracy of the PMMH procedures in estimating the smoothing expectation $\mathbb{E}[\mathcal{V}_n(X_{0:n})|y_{0:n}]$ is measured using $e^N_n$ and $e^{N,\epsilon}_n$. As above, these errors are calculated as the (dimension-averaged) $\mathbb{L}_1-$errors of the PMMH estimates under the exact and ABC HMM, respectively. All results are repeated over 50 independent runs.

\subsubsection{Results}

Our results are displayed in Figures \ref{fig:pmmh_smooth}-\ref{fig:pmmh_sigy}.
In Figure \ref{fig:pmmh_smooth}, we can observe the accuracy of PMMH estimation of the smoothed additive functional, using SMC updates both with and without forward smoothing, under both the exact HMM and its ABC approximation.
%In Figure \ref{fig:pmmh_smooth}, we can observe the estimation of the smoothed additive functional for both exact SMC with and without forward only smoothing and similarly for the ABC approach.
%First with regards to the estimation of smoothed additive functionals (Figure \ref{fig:pmmh_smooth}).
Here we observe the expected pattern; the use of forward only smoothing in the PMMH update scheme significantly enhances estimative accuracy for roughly the same computational cost - the accuracy is better and the variance lower. When using forward smoothing in the SMC update mechanism, we further observe that the ABC HMM can be targeted with reasonable accuracy. Consider the effect of increasing $N$ on the errors in Figure \ref{fig:pmmh_smooth}. Interestingly, the improvement in estimation is more evident when using forward only smoothing, even though one expects a PMMH algorithm with more particles to mix better (see e.g.~\cite{andrieu_pseudo}) and thus the estimation to be most likely improved.

In terms of the estimation of parameters, we consider Figures \ref{fig:pmmh_sigx} and \ref{fig:pmmh_sigy}. Here, we are mainly concerned
with the quality of parameter estimation under the ABC HMM without forward smoothing - the forward smoothing cannot contribute anything to parameter estimation here.
The accuracy of the ABC is, in general quite biased by up-to 40\% of the parameter values. The variance is also quite substantial relative to the exact approach.

\begin{figure}[!h]\centering
\vspace{-5cm}
\includegraphics[width=14cm,height=20cm,angle=0]{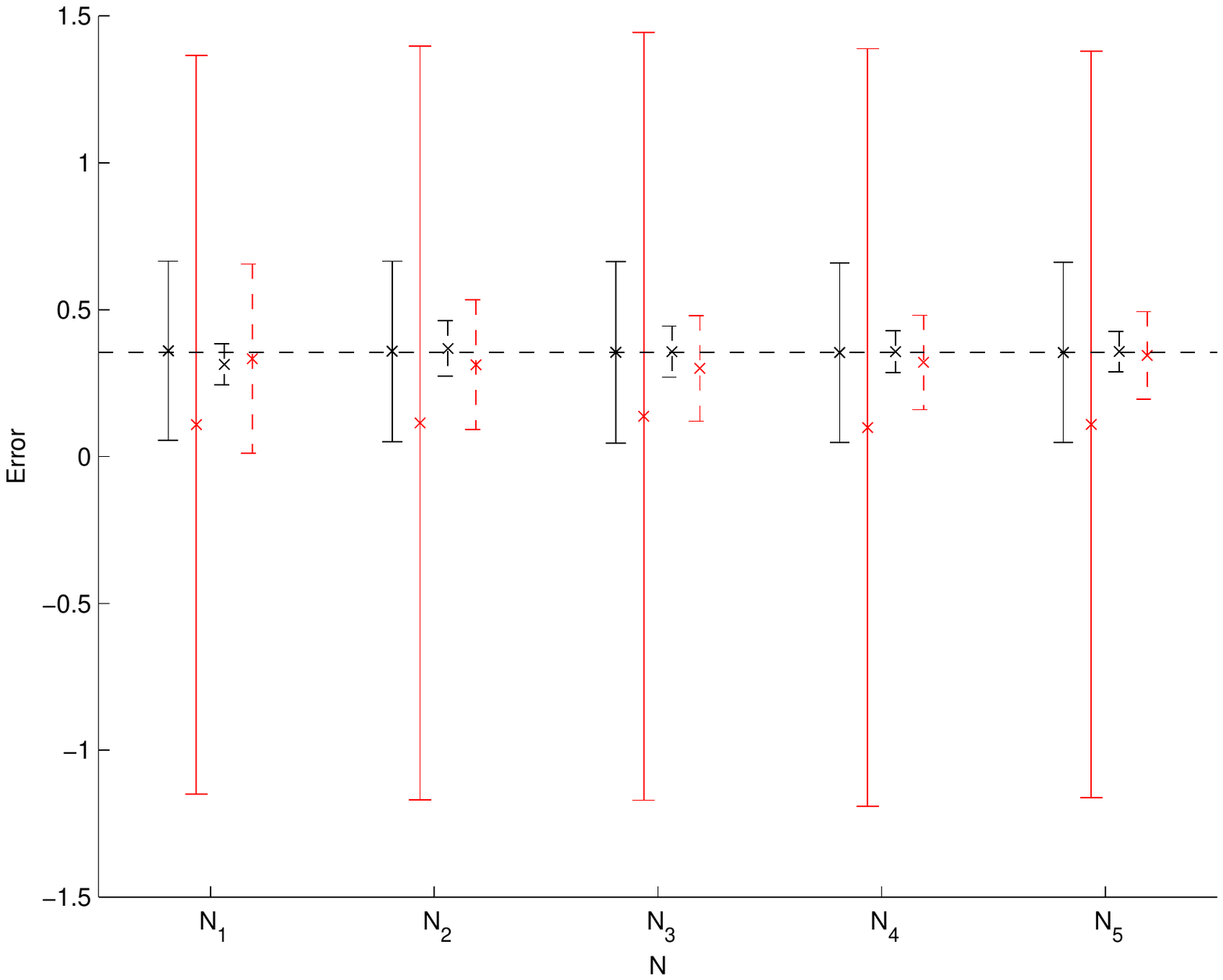}
\vspace{-4.5cm}
\caption{Standard errors of the smoothing errors associated with the estimates 
of the mean state obtained over 50 independent implementations. 
The PMCMC with exact SMC is in black, the ABC in red. The dotted lines indicate the usage of
forward only smoothing.
The dotted horizontal line is the estimated true value.}
\label{fig:pmmh_smooth}
\end{figure}

\begin{figure}[!h]\centering
\vspace{-5cm}
\includegraphics[width=14cm,height=20cm,angle=0]{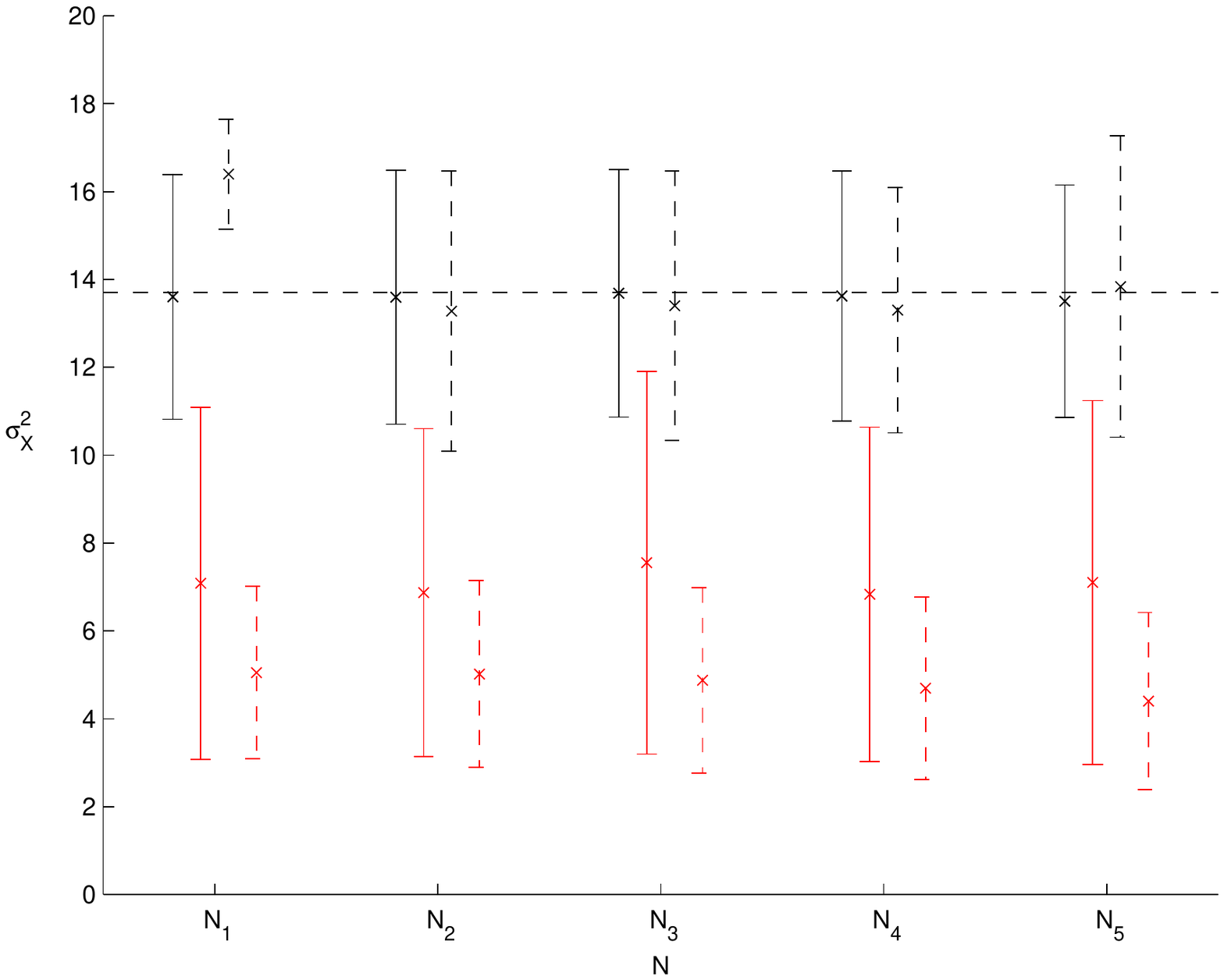}
\vspace{-4.5cm}
\caption{Standard errors of the smoothing errors associated with the estimates 
of $\sigma^2_X$ obtained over 50 independent implementations. 
The PMCMC with exact SMC is in black, the ABC in red. The dotted lines indicate the usage of
forward only smoothing.
The dotted horizontal line is the estimated true value.}
\label{fig:pmmh_sigx}
\end{figure}

\begin{figure}[!h]\centering
\vspace{-5cm}
\includegraphics[width=14cm,height=20cm,angle=0]{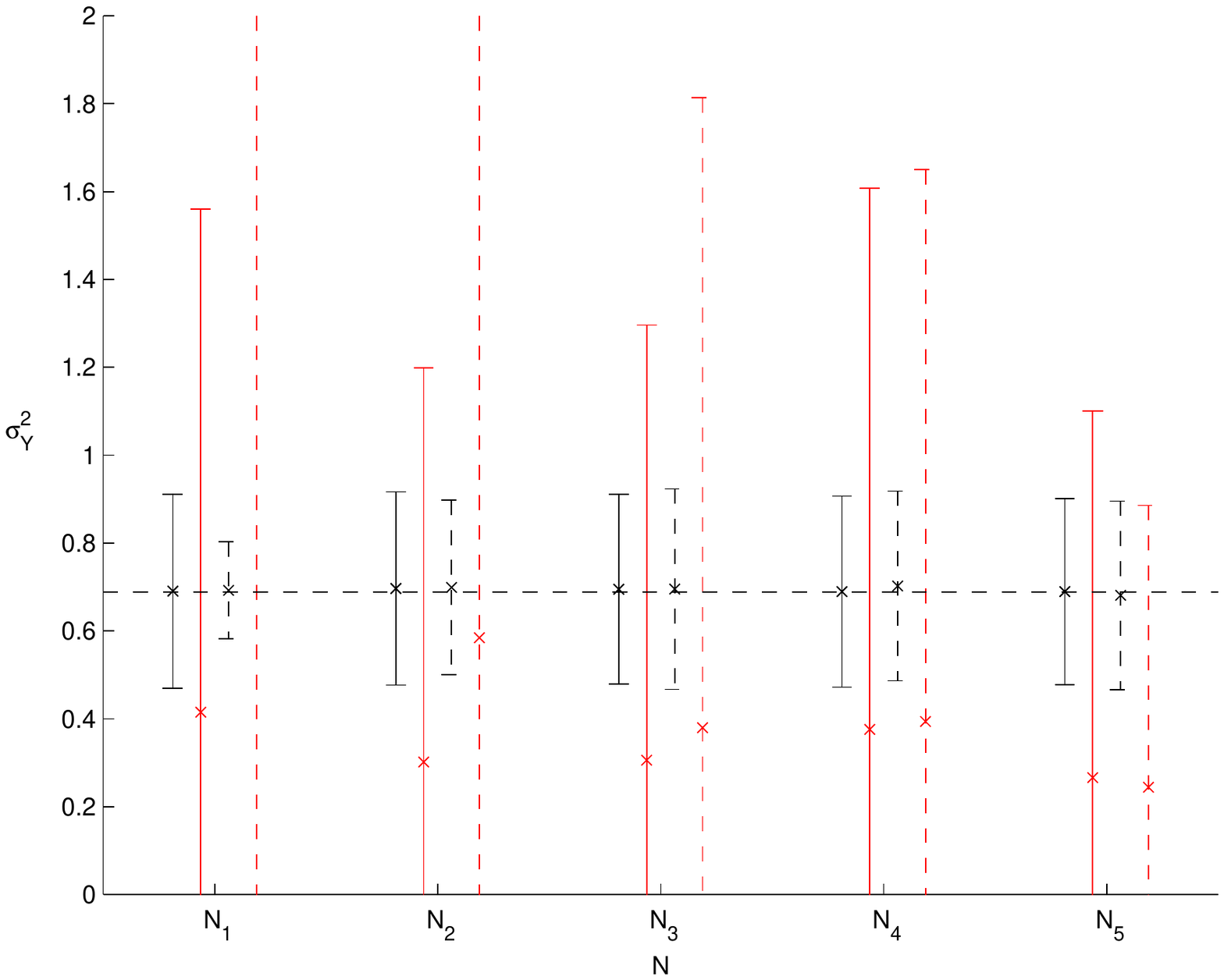}
\vspace{-4.5cm}
\caption{Standard errors of the smoothing errors associated with the estimates 
of $\sigma^2_Y$ obtained over 50 independent implementations. 
The PMCMC with exact SMC is in black, the ABC in red. The dotted lines indicate the usage of
forward only smoothing.
The dotted horizontal line is the estimated true value.}
\label{fig:pmmh_sigy}
\end{figure}

\subsection{Conclusions}

On the basis of our numerical study we can tentatively conclude the following. For smoothing:
\begin{itemize}
\item{In higher dimensions, ABC, in terms of accuracy, is competitive with using standard SMC (even if the model is analytically tractable);}
\item{For ABC with $\phi\big(\frac{u-y}{\epsilon}\big)$ specified as an indicator function, one would prefer to use the RSMC procedure over an SMC procedure with dynamic resampling.}
\end{itemize}
For the use of PMMH for performing batch parameter estimation, it would appear that, for moderate length time series, using forward only smoothing is not necessarily useful. If one is interested in %\emph{both} parameter estimation and 
the estimation of smoothed additive functionals, however, the use of forward smoothing can provide significant improvements (for the same computational cost) in estimative accuracy when compared to the PMMH procedure in \cite{andrieu}.
The ABC procedure produces parameter estimates which are perhaps more biased than estimation of smoothed additive functionals, but this is also linked to the
fact that the estimation method used is focussed on the latter quantities.
These conclusions, of course, cannot be comprehensive as they are model and quantity (w.r.t.~estimation) dependent.
However, we have seen similar trends in different examples or different parameter settings for the same model.

\section{Summary}\label{sec:summary}

In this article we have investigated smoothing and static parameter estimation for HMMs with intractable likelihoods. We have constructed SMC and PMCMC based-solutions for ABC approximations and investigated the bias 
associated to our procedure. There are several extensions to the work that has been considered here.
From the perspective of parameter estimation, we have only considered batch estimation by using PMCMC. In many practical problems, one is often interested in performing statistical inference as samples arrive online. We are currently investigating methodology
for this problem in \cite{ehrlich,yidrlim} and the theoretical and empirical work here
is of great relevance in these latter ideas; in particular when applying the online EM algorithm as considered in \cite{yidrlim}.
%In terms of our bias result, we have used extremely restrictive assumptions, which will typically only hold when the hidden Markov chain and observations  are constrained to a compact or finite state-space. However, there has been some progress on relaxing these assumptions in the context of SMC approximations of the prediction distribution (in the context of filtering) \cite{whiteley}.
%We are currently working on doing the same for SMC approximations of smoothing and then subsequently for ABC \cite{jacob}. 
In this article we have focussed upon the forward-only smoothing technique in \cite{dds2}, however, this is not the only possibility;
one can also investigate the ideas in \cite{Briers_2010} in the context of ABC. In particular, the relative performance of these procedures is of interest.

\subsubsection*{Acknowledgements} 

The first author was supported by an EPSRC grant. The second author was supported by an MOE grant. The second and fourth authors acknowledge assistance from the LMS research in pairs grant. The second and third authors acknowledge assistance from an EPSRC platform grant EP/I019111/1.

\appendix

\section{SMC Error}\label{app:smc_bound}

We give an analysis of the self-normalized estimate that was not the objective in \cite{dds2} and is explicit in $\epsilon$. To that end, we introduce the following notations, to keep a consistent notation with \cite{dds1,dds2} on which our analysis relies. We set:
\begin{eqnarray*}
H((x,u),(x,u')) & = & g(x',u') f(x,x')\\
G_{n,\epsilon}(u) & = & \phi\bigg(\frac{u-y_n}{\epsilon}\bigg).
\end{eqnarray*}
To avoid notational overload, we will simply write $H(x,x')$, $x\in\mathbb{R}^{d_x}\times\mathbb{R}^{d_y}=:E$
and $G_{n,\epsilon}(x)$, $x\in\mathbb{R}^{d_x}\times\mathbb{R}^{d_y}$ (despite the independence of $G_{n,\epsilon}$ on only $\mathbb{R}^{d_x}$).

We consider the approximation of the path measure $\mathbb{Q}_{n,\epsilon}(\mathcal{V}_n)$, which is by definition:
$$
\mathbb{Q}_{n,\epsilon}(\mathcal{V}_n) := \frac{1}{\gamma_{n,\epsilon}(1)}\int_{E^{n+1}} 
\bigg[\prod_{p=0}^{n-1} G_{p,\epsilon}(x_p)\bigg] \mathcal{V}_n(x_{0:n}) \eta_0(x_0) \prod_{p=1}^n H(x_{p-1},x_p) dx_{0:p}
$$
where 
$$
\gamma_{n,\epsilon}(f)= \int_{E^{n+1}} 
\bigg[\prod_{p=0}^{n-1} G_{p,\epsilon}(x_p)\bigg]f(x_n) \eta_0(x_0) \prod_{p=1}^n H(x_{p-1},x_p) dx_{0:p} 
$$
and $f:E\rightarrow\mathbb{R}$. Recall for additive functionals $\mathcal{V}_n(x_{0:n})=\sum_{p=0}^n v_p(x_p)$
$$
V_{n,\epsilon}^N(x) = v_n(x) + \sum_{i=1}^N \frac{G_{n-1,\epsilon}(x_{n-1}^i)H_{n}(x_{n-1}^i,x)}
{\sum_{j=1}^N G_{n-1,\epsilon}(x_{n-1}^j)H_{n}(x_{n-1}^j,x)}V_{n-1,\epsilon}^N(x_{n-1}^i)
$$
with $V_{0,\epsilon}^N=v_0$; c.f.~\eqref{eq:v_update}. We remind the reader that $v_p$ is a function on $\mathbb{R}^{d_x}$ only.

\begin{theorem}\label{theo:smc_error}
Assume (A\ref{hyp:A}). Then for any $1\leq p <+\infty$, there exist a $a_p<+\infty$ such that for any $\epsilon>0$, $N\geq 1$, $n\geq 1$, $y_{0:n}$:
$$
\bigg\|\sum_{i=1}^N\frac{G_{n,\epsilon}(x_n^i)}{\sum_{j=1}^N G_{n,\epsilon}(x_n^j)} V_{n,\epsilon}^N(x_n^i)
 - \mathbb{E}^{\epsilon}[\mathcal{V}_n(X_{0:n})|y_{0:n}] \bigg\|_p \leq \frac{a_p(n+1)\overline{\delta}(\epsilon)}{\sqrt{N}}
$$
where $\overline{\delta}(\epsilon)=\max\{\delta(\epsilon)\max\{\delta(\epsilon)^2,\delta(\epsilon)^4\}, \delta(\epsilon)^{3}\}$, with $\delta(\epsilon)$ as in (A\ref{hyp:A}) (\ref{hyp:smooth_abc}) and $\mathbb{E}^{\epsilon}[\cdot|y_{0:n}]$ is the expectation w.r.t.~the joint ABC smoothing distribution.
\end{theorem}

\begin{proof}
Consider the decomposition:
$$
\sum_{i=1}^N\frac{G_{n,\epsilon}(x_n^i)}{\sum_{j=1}^N G_{n,\epsilon}(x_n^j)} V_{n,\epsilon}^N(x_n^i) - \frac{\mathbb{Q}_{n,\epsilon}(G_{n,\epsilon} \mathcal{V}_n)}{\eta_{n,\epsilon}(G_{n,\epsilon})} =  
$$
$$
\frac{\frac{1}{N}\sum_{i=1}^N G_{n,\epsilon}(x_n^i)V_{n,\epsilon}^N(x_n^i)}{\eta_{n,\epsilon}(G_{n,\epsilon})\frac{1}{N}\sum_{j=1}^N G_{n,\epsilon}(x_n^j)}\bigg(\eta_{n,\epsilon}(G_{n,\epsilon}) -\frac{1}{N}\sum_{j=1}^N G_{n,\epsilon}(x_{n}^j)\bigg) 
$$
\begin{equation}
+ \frac{1}{\eta_{n,\epsilon}(G_{n,\epsilon})}
\bigg(\frac{1}{N}\sum_{i=1}^N G_{n,\epsilon}(x_n^i)V_{n,\epsilon}^N(x_n^i) - \mathbb{Q}_{n,\epsilon}(G_{n,\epsilon} \mathcal{V}_n)\bigg)
\label{eq:decomp}
\end{equation}
where we recall the normalized $n-$time marginal $\eta_{n,\epsilon}(f) = \gamma_{n,\epsilon}(f)/\gamma_{n,\epsilon}(1)$. Note that 
\begin{eqnarray*}
\frac{\mathbb{Q}_{n,\epsilon}(G_{n,\epsilon} \mathcal{V}_n)}{\eta_{n,\epsilon}(G_{n,\epsilon})} & = & \frac{1}{\gamma_{n,\epsilon}(G_{n,\epsilon})}\int_{E^{n+1}} 
\bigg[\prod_{p=0}^{n} G_{p,\epsilon}(x_p)\bigg] \mathcal{V}_n(x_{0:n}) \eta_0(x_0) \prod_{p=1}^n H(x_{p-1},x_p)dx_{0:p}\\
& = &  \mathbb{E}^{\epsilon}[\mathcal{V}_n(X_{0:n})|y_{0:n}]
\end{eqnarray*}
which is the quantity of interest.

To consider an $\mathbb{L}_p-$analysis, we can split the two terms in \eqref{eq:decomp}
via Minkowski. We consider the first term:
$$
\bigg\|\frac{\frac{1}{N}\sum_{i=1}^N G_{n,\epsilon}(x_n^i)V_{n,\epsilon}^N(x_n^i)}{\eta_{n,\epsilon}(G_{n,\epsilon})
\frac{1}{N}\sum_{j=1}^N G_{n,\epsilon}(x_n^j)}\bigg(\eta_{n,\epsilon}(G_{n,\epsilon}) -\frac{1}{N}\sum_{j=1}^N G_{n,\epsilon}(x_n^j)\bigg)\bigg\|_p
$$
Now, one has: $V_{n,\epsilon}^N(x_n^i)\leq \sum_{p=0}^n \|v_p\|$. This is proved by induction, the initialization with $n=0$ being obvious, assuming for $n-1$, one has:
$$
V_{n,\epsilon}^N(x) = v_n(x) + \sum_{i=1}^N \frac{G_{n-1,\epsilon}(x_{n-1}^i)H_{n}(x_{n-1}^i,x)}
{\sum_{j=1}^N G_{n-1,\epsilon}(x_{n-1}^j)H_{n}(x_{n-1}^j,x)}V_{n-1,\epsilon}^N(x_{n-1}^i) \leq
v_n(x) + \sum_{p=0}^{n-1} \|v_p\|
$$
and one easily concludes. Thus, it follows:
$$
\bigg\|\frac{\frac{1}{N}\sum_{i=1}^N G_{n,\epsilon}(x_n^i)V_{n,\epsilon}^N(x_n^i)}{\eta_{n,\epsilon}(G_{n,\epsilon})
\frac{1}{N}\sum_{j=1}^N G_{n,\epsilon}(x_n^j)}\bigg(\eta_{n,\epsilon}(G_{n,\epsilon}) -\sum_{j=1}^N G_{n,\epsilon}(x_n^j)\bigg)\bigg\|_p
$$
$$
\leq (n+1)\overline{v}\underline{\alpha}(\epsilon)^{-1}\bigg\|\eta_{n,\epsilon}(G_{n,\epsilon}) -\frac{1}{N}\sum_{j=1}^N G_{n,\epsilon}(x_n^j) \bigg\|_p.
$$
By Theorem 7.4.4 of \cite{delmoral} that there exist some
$a_p<\infty$ such that
$$
\bigg\|\bigg(\eta_{n,\epsilon}(G_{n,\epsilon}) -\frac{1}{N}\sum_{j=1}^N G_{n,\epsilon}(x_n^j)\bigg)\bigg\|_p \leq \frac{a_p\delta(\epsilon)^2\overline{\alpha}(\epsilon)}{\sqrt{N}}
$$
hence
$$
\bigg\|\frac{\frac{1}{N}\sum_{i=1}^N G_{n,\epsilon}(x_n^i)V_{n,\epsilon}^N(x_n^i)}{\eta_{n,\epsilon}(G_{n,\epsilon})
\frac{1}{N}\sum_{j=1}^N G_{n,\epsilon}(x_n^j)}\bigg(\eta_{n,\epsilon}(G_{n,\epsilon}) -\frac{1}{N}\sum_{j=1}^N G_{n,\epsilon}(x_n^j)\bigg)\bigg\|_p \leq 
\frac{a_p (n+1)\delta(\epsilon)^3}{\sqrt{N}}
$$
for some $a_p<\infty$ that does not depend upon $n$ or $\epsilon$.
To deal with the second term in  \eqref{eq:decomp} one can use Lemma \ref{lem:updated_lp_bound} along with (A\ref{hyp:A}) (\ref{hyp:smooth_abc}) to conclude.
\end{proof}

\subsection{Technical Result}

We provide a proof of Lemma \ref{lem:updated_lp_bound}.
To that end, introduce the operator:
$$
D_{p,n,\epsilon}^N(V_n)(x_p) = \int \mathcal{M}_{p,\epsilon}^N(x_p,d_{0:p-1})\mathcal{Q}_{p,n,\epsilon}
(x_p,dx_{p:n})V_n(x_n)
$$ 
where 
\begin{eqnarray*}
\mathcal{M}_{p,\epsilon}^N(x_p,dx_{0:p-1}) & =  & \prod_{q=1}^p M_{q,\eta_{q-1}^N,\epsilon}(x_q,dx_{q-1})\\
\mathcal{Q}_{p,n,\epsilon}(x_p,dx_{p:n}) & = & \prod_{q=p+1}^n Q_{q,\epsilon}(x_{q-1},dx_q)\\
M_{q,\eta_{q-1}^N,\epsilon}(x_q,dx_{q-1}) & = & \frac{\eta_{q-1}^N(dx_{q-1})G_{q-1,\epsilon}(x_{q-1})
H_q(x_{q-1},x_q)}{\eta_{q-1}^N(G_{q-1,\epsilon}H_q(\cdot,x_q))}\\
Q_{q,\epsilon}(x_{q-1},dx_q) & = & G_{q-1,\epsilon}(x_{q-1})H_q(x_{q-1},x_q)dx_q\\
\eta_{q-1}^N(dx_{q-1}) & = & \frac{1}{N}\sum_{i=1}^N \delta_{x_{q-1}^i}(dx_{q-1})
\end{eqnarray*}
where all the conventions of \cite{dds1} are preserved. In the empirical measure in the final line, one considers the mutated particles.
We also use the convention $Q_{p,n,\epsilon}=Q_{p+1,\epsilon}\dots Q_{n,\epsilon}$. Note for the backward kernel, when considering the filter $\hat{\eta}_{q-1,\epsilon}$ one can write
\begin{equation*}
M_{q,\hat{\eta}_{q-1,\epsilon},\epsilon}(x_q,dx_{q-1}) = 
\frac{\hat{\eta}_{q-1,\epsilon}(dx_{q-1})H_q(x_{q-1},x_q)}{\hat{\eta}_{q-1,\epsilon}(H_q(\cdot,x_q))}.
\end{equation*}

\begin{lem}\label{lem:updated_lp_bound}
Assume (A\ref{hyp:A}). Then for any $1\leq p <+\infty$ there exist a $a_p<+\infty$ such that for any $\epsilon>0$, $N\geq 1$, $n\geq 1$, $y_{0:n}$:
$$
\bigg\|\frac{1}{N}\sum_{i=1}^N G_{n,\epsilon}(x_n^i)V_{n,\epsilon}^N(x_n^i) - \mathbb{Q}_{n,\epsilon}(G_{n,\epsilon} \mathcal{V}_n)\bigg\|_p \leq \frac{a_p (n+1)\widetilde{\delta}(\epsilon)\overline{\alpha}(\epsilon)}{\sqrt{N}}
$$
where $\widetilde{\delta}(\epsilon)=\max\{\delta(\epsilon)^2, \delta(\epsilon)^{4}\}$, with $\delta(\epsilon),\overline{\alpha}(\epsilon)$ as in (A\ref{hyp:A}) (\ref{hyp:smooth_abc}).
\end{lem}

\begin{proof}
The proof of this result follows the proof of Theorem 3.2 of \cite{dds1}. The complications are the control of the oscillations of
$P_{p,n,\epsilon}^N(G_{n,\epsilon}\mathcal{V}_n)=D_{p,n,\epsilon}^N(G_{n,\epsilon}\mathcal{V}_n)/D_{p,n}^N(1)$ with a different function as well as obtaining a rate w.r.t.~$\epsilon$. It is remarked that the application of the Kintchine-type inequality in Theorem 3.2 of \cite{dds1} will not add any dependence upon $\epsilon$.

Using the definition of $D_{p,n,\epsilon}^N$ one has
$$
D_{p,n,\epsilon}^N(G_{n,\epsilon}\mathcal{V}_n) = Q_{p,n,\epsilon}(G_{n,\epsilon})\sum_{q=0}^p M_{p,\eta_{p-1}^N,\epsilon}\dots M_{q+1,\eta_{q}^N,\epsilon}(v_q) + \sum_{q=p+1}^n Q_{p.q,\epsilon}(v_q Q_{q,n,\epsilon}(G_{n,\epsilon})).
$$
Thus it follows that:
\begin{equation}
P_{p,n,\epsilon}^N(G_{n,\epsilon}\mathcal{V}_n) = \frac{Q_{p,n,\epsilon}(G_{n,\epsilon})}{Q_{p,n,\epsilon}(1)} \sum_{q=0}^{p-1} M_{p,\eta_{p-1}^N,\epsilon}\dots M_{q+1,\eta_{q}^N,\epsilon}(v_q) + \sum_{q=p}^n \frac{S_{p,q,\epsilon}(\overline{Q}_{q,n,\epsilon}(G_{n,\epsilon})v_q)}{S_{p,q,\epsilon}(\overline{Q}_{q,n,\epsilon}(1))}
\label{eq:p_decomp}
\end{equation}
where $S_{p,q,\epsilon}(v) = Q_{p,q,\epsilon}(v)/Q_{p,q,\epsilon}(1)$ and $\overline{Q}_{q,n,\epsilon}(v)=Q_{q,n,\epsilon}(v)/\eta_{q,\epsilon}(Q_{q,n,\epsilon}(1))$. To deal with oscillations of $P_{p,n,\epsilon}^N(G_{n,\epsilon}\mathcal{V}_n)$ we consider both sums separately.

We being with the first term on the R.H.S.~of\eqref{eq:p_decomp}. In particular the difference with arguments $x$ and $y$:
$$
\frac{Q_{p,n,\epsilon}(G_{n,\epsilon})(x)}{Q_{p,n,\epsilon}(1)(x)} \sum_{q=0}^{p-1} M_{p,\eta_{p-1}^N,\epsilon}\dots M_{q+1,\eta_{q}^N,\epsilon}(v_q)(x) - 
\frac{Q_{p,n,\epsilon}(G_{n,\epsilon})(y)}{Q_{p,n,\epsilon}(1)(y)} \sum_{q=0}^{p-1} M_{p,\eta_{p-1}^N,\epsilon}\dots M_{q+1,\eta_{q}^N,\epsilon}(v_q)(y).
$$
Then, treating the summands, one has
$$
\frac{Q_{p,n,\epsilon}(G_{n,\epsilon})(x)}{Q_{p,n,\epsilon}(1)(x)}\bigg[M_{p,\eta_{p-1}^N,\epsilon}\dots M_{q+1,\eta_{q}^N,\epsilon}(v_q)(x) -M_{q+1,\eta_{q}^N,\epsilon}(v_q)(y)\bigg] +
$$
$$
M_{q+1,\eta_{q}^N,\epsilon}(v_q)(y)\bigg[S_{p,n,\epsilon}(G_{n,\epsilon})(x) - S_{p,n,\epsilon}(G_{n,\epsilon})(y)\bigg]
$$
which is clearly upper-bounded by
\begin{equation}
2\|G_{n,\epsilon}\|\|v_q\|[\beta(M_{p,\eta_{p-1}^N,\epsilon}\dots M_{q+1,\eta_{q}^N,\epsilon}) + \beta(S_{p,n,\epsilon})].
\label{eq:bound1}
\end{equation}

Now consider the second term on the R.H.S.~of \eqref{eq:p_decomp}. In particular, for each summand, one has after subtracting the function in $y$ from that in $x$
$$
\frac{S_{p,q,\epsilon}(\overline{Q}_{q,n,\epsilon}(G_{n,\epsilon})v_q)(x) - S_{p,q,\epsilon}(\overline{Q}_{q,n,\epsilon}(G_{n,\epsilon})v_q)(y)}
{S_{p,q,\epsilon}(\overline{Q}_{q,n,\epsilon}(1)(x)} 
$$
$$
+ S_{p,q,\epsilon}(\overline{Q}_{q,n,\epsilon}(G_{n,\epsilon})v_q)(y)
\bigg[\frac{S_{p,q,\epsilon}(\overline{Q}_{q,n,\epsilon}(1)(y)-S_{p,q,\epsilon}(\overline{Q}_{q,n,\epsilon}(1)(x)}{S_{p,q,\epsilon}(\overline{Q}_{q,n,\epsilon}(1)(y)
S_{p,q,\epsilon}(\overline{Q}_{q,n,\epsilon}(1)(x)}\bigg]
$$
Dealing with the two terms separately, one can easily show that the first term is upper-bounded by
$
2b_{q,n,\epsilon}^2 \|v_q\|\|G_{n,\epsilon}\|\beta(S_{p,q,\epsilon})
$
where $b_{q,n,\epsilon}=\sup_{x,y}Q_{p,n,\epsilon}(1)(x)/Q_{p,n,\epsilon}(1)(y)$. Similarly using trivial manipulations, one can also show that the second term is upper-bounded by same term, yielding the upper-bound 
\begin{equation}
4b_{q,n,\epsilon}^2 \|v_q\|\|G_{n,\epsilon}\|\beta(S_{p,q,\epsilon})\label{eq:bound2}
\end{equation}

Combining the bounds \eqref{eq:bound1}-\eqref{eq:bound2} one can deduce that:
\begin{eqnarray*}
\textrm{Osc}(P_{p,n,\epsilon}^N(G_{n,\epsilon}\mathcal{V}_n)) & \leq & 2\|G_{n,\epsilon}\|\sum_{q=0}^{p-1}\|v_q\|[\beta(M_{p,\eta_{p-1}^N,\epsilon}\dots M_{q+1,\eta_{q}^N,\epsilon}) + \beta(S_{p,n,\epsilon})] \\ & & + 
4\|G_{n,\epsilon}\|\sum_{q=0}^{p-1}\|v_q\|b_{q,n,\epsilon}^2 \beta(S_{p,q,\epsilon}).
\end{eqnarray*}
Thus, one has, via the proof of Theorem 3.2 of \cite{dds1}:
$$
\bigg\|\frac{1}{N}\sum_{i=1}^N G_{n,\epsilon}(x_n^i)V_{n,\epsilon}^N(x_n^i) - \mathbb{Q}_{n,\epsilon}(G_{n,\epsilon} \mathcal{V}_n)\bigg\|_p \leq a_p \sum_{p=0}^n b_{p,n,\epsilon}^2 c_{p,n,\epsilon}^N
$$
where
$$
c_{p,n,\epsilon}^N = \mathbb{E}\bigg[2\|G_{n,\epsilon}\|\sum_{q=0}^{p-1}\|v_q\|[\beta(M_{p,\eta_{p-1}^N,\epsilon}\dots M_{q+1,\eta_{q}^N,\epsilon}) + \beta(S_{p,n,\epsilon})] + 
4\|G_{n,\epsilon}\|\sum_{q=0}^{p-1}\|v_q\|b_{q,n,\epsilon}^2 \beta(S_{p,q,\epsilon})\bigg].
$$

To complete the proof we need to consider the term $c_{p,n,\epsilon}^N$, to that end, we quote the following bounds which follow from (A\ref{hyp:A}); see \cite{dds1} and the citations therein for details:
$$
b_{p,n,\epsilon} \leq \delta(\epsilon)\rho^4 \quad \beta(S_{p,q,\epsilon}) \leq (1-\delta(\epsilon)^{-1}\rho^{-8})^{q-p} \quad \beta(M_{p,\eta_{p-1}^N,\epsilon}\dots M_{q+1,\eta_{q}^N,\epsilon}) \leq (1-\rho^{-8})^{p-q}.
$$
First consider the expression
$$
2\|G_{n,\epsilon}\|\sum_{p=0}^nb_{p,n,\epsilon}^2\sum_{q=0}^{p-1}\|v_q\|[\beta(M_{p,\eta_{p-1}^N,\epsilon}\dots M_{q+1,\eta_{q}^N,\epsilon}) + \beta(S_{p,n,\epsilon})]
$$
which is upper-bounded by
$$
2\overline{\alpha}(\epsilon)(\delta(\epsilon)\rho^4)^2\overline{v}\sum_{p=0}^n
\sum_{q=0}^{p-1}[(1-\rho^{-8})^{p-q}+(1-\delta(\epsilon)^{-1}\rho^{-8})^{n-p}]
$$
with $\overline{v}$ as in (A\ref{hyp:A}) (\ref{hyp:function}).
By standard manipulations, this is upper-bounded by
$$
C \overline{\alpha}(\epsilon)\delta(\epsilon)^2(n+1)
$$
for $C<\infty$ that does not depend upon $n$ or $\epsilon$.
Second the expression
$$
\sum_{p=0}^n b_{p,n,\epsilon}^2 4\|G_{n,\epsilon}\|\sum_{q=0}^{p-1}\|v_q\|b_{q,n,\epsilon}^2 \beta(S_{p,q,\epsilon})
$$
which is upper-bounded by
$$
4\overline{\alpha}(\epsilon)(\delta(\epsilon)\rho^4)^4\overline{v}\sum_{p=0}^n\sum_{q=p}^n
(1-\delta(\epsilon)^{-1}\rho^{-8})^{q-p}.
$$
Again, by standard manipulations one can upper-bound this latter expression by
$$
C \overline{\alpha}(\epsilon)\delta(\epsilon)^4(n+1)
$$
for $C<\infty$ that does not depend upon $n$ or $\epsilon$;
we can now conclude.
\end{proof}

\section{ABC Error}\label{app:abc_error}

Below we will repeatedly apply Theorem 2 of \cite{jasra}. This can also be established under (A\ref{hyp:A}) for smoothed ABC; the proof is omitted and follows the description in \cite{jasra}.
Recall that the ABC approximation of the joint smoothing density is:
$$
\hat{\eta}_{n,\epsilon}(x_{0:n}) = \frac{[\prod_{i=0}^n \int_{\mathbb{R}^{d_y}}\phi(\frac{u-y_i}{\epsilon})g(x_i,u)du] \eta_0(x_0)\prod_{i=1}^n f(x_{i-1},x_i)}{\int_{\mathbb{R}^{(n+1)d_x}}[\prod_{i=0}^n \int_{\mathbb{R}^{d_y}}\phi(\frac{u-y_i}{\epsilon})g(x_i,u)du] \eta_0(x_0)\prod_{i=1}^n f(x_{i-1},x_i) dx_{0:n}}.
$$
It is stressed that the analysis here is performed by integrating the auxiliary data, whilst the SMC analysis works on the joint space of auxiliary data and hidden state.

\begin{proof}[Proof of Theorem \ref{theo:abc_error}]
We have
$$
|\mathbb{E}[\mathcal{V}_n(X_{0:n})|y_{0:n}]-\mathbb{E}^{\epsilon}[\mathcal{V}_n(X_{0:n})|y_{0:n}]|
= |\sum_{p=0}^n \int v_p(x_p)[\hat{\eta}_n(x_{p:n})-\hat{\eta}_{n,\epsilon}(x_{p:n})]dx_{p:n}|
$$
where $\hat{\eta}_{p,\epsilon}(x_{p:n})$ and $\hat{\eta}_p(x_{p:n})$ are the ABC and true smoothers. Using the backward representation of the smoothers, one has the decomposition of the R.H.S.:
\begin{equation}
|\sum_{p=0}^n \int v_p(x_p)[\hat{\eta}_n(x_{p:n})-\hat{\eta}_{n,\epsilon}(x_{p:n})]dx_{p:n}|
=
|\sum_{p=0}^n \hat{\eta}_n M_{n:p+1,\hat{\eta}_{n-1:p}}(v_p) - \hat{\eta}_n^{\epsilon} M_{n:p+1,\hat{\eta}_{n-1:p,\epsilon},\epsilon}(v_p)|\label{eq:first_bias_decomp}
\end{equation}
where
\begin{eqnarray*}
\hat{\eta}_n M_{n:p+1,\hat{\eta}_{n-1:p}}(v_p) & = & \int \hat{\eta}_n(x_n) \prod_{q=p+1}^{n} M_{q,\hat{\eta}_{q-1}}(x_q,dx_{q-1}) v_p(x_p) dx_n\\
\hat{\eta}_{n,\epsilon} M_{n:p+1,\hat{\eta}_{n-1:p,\epsilon},\epsilon}(v_p) & = & \int \hat{\eta}_{n,\epsilon}(x_n) \prod_{q=p+1}^{n} M_{q,\hat{\eta}_{q-1,\epsilon},\epsilon}(x_q,dx_{q-1}) v_p(x_p) dx_n
\end{eqnarray*}
with the backward kernels:
\begin{eqnarray}
M_{q,\hat{\eta}_{q-1}}(x_q,dx_{q-1}) & = & \frac{\hat{\eta}_{q-1}(x_{q-1})f(x_{q-1},x_q)}{\hat{\eta}_{q-1}(f(\cdot,x_q))}dx_{q-1}\label{eq:backward_true}\\
M_{q,\hat{\eta}_{q-1,\epsilon},\epsilon}(x_q,dx_{q-1}) & = & \frac{\hat{\eta}_{q-1,\epsilon}(x_{q-1})f(x_{q-1},x_q)}
{\hat{\eta}_{q-1,\epsilon}(f(\cdot,x_q))}dx_{q-1}
\label{eq:backward_abc}.
\end{eqnarray}
We will drop the $\eta$ subscripts for the remainder of the proof.

One can now adopt a telescoping sum decomposition for each summand of the R.H.S.~of \eqref{eq:first_bias_decomp}:
$$
\sum_{s=1}^{n-p} \bigg(\hat{\eta}_n M_{n:n-s,\epsilon}[M_{n-s+1}-M_{n-s+1,\epsilon}](M_{n-s:p+1}(v_p))\bigg) + [\hat{\eta}_n -\hat{\eta}_{n,\epsilon}](M_{n:p+1,\epsilon}(v_p))
$$
For the second term, one can use Theorem 2 of \cite{jasra}. Thus concentrating on the summands in the first term we have
\begin{eqnarray*}
 \hat{\eta}_n M_{n:n-s,\epsilon}[M_{n-s+1}-M_{n-s+1,\epsilon}](M_{n-s:p+1}(v_p))
&\leq & \|\hat{\eta}_n M_{n:n-s,\epsilon}[M_{n-s+1}-M_{n-s+1,\epsilon}]\|_{TV}\times\\ & & \big(\prod_{q=n-s}^{p+1} \beta(M_{q})\big) \textrm{Osc}(v_p)
\end{eqnarray*}
via (A\ref{hyp:A}) (\ref{hyp:density_control}), (\ref{hyp:function}) and Lemma \ref{prop:backward_control} it clearly follows that there exist a $C<\infty$ and $\xi\in(0,1)$ which do not depend upon $n$, $y_{0:n}$ $\epsilon$ such that
$$
 \hat{\eta}_n M_{n:n-s,\epsilon}[M_{n-s+1}-M_{n-s+1,\epsilon}](M_{n-s:p+1}(v_p))
 \leq C \epsilon \xi^{p+s+2-n}.
$$
As a result, we have
$$
|\mathbb{E}[\mathcal{V}_n(X_{0:n})|y_{0:n}]-\mathbb{E}^{\epsilon}[\mathcal{V}_n(X_{0:n})|y_{0:n}]|
\leq C\sum_{p=0}^n\bigg[\big(\sum_{s=1}^{n-p}\xi^{p+s+2-n}\big) + 1\bigg]\epsilon
$$
for $C<\infty$ that does not depend upon $n$, $y_{0:n}$ $\epsilon$. Elementary manipulations allow us to conclude. 
\end{proof}

\subsection{Technical Result}

\begin{lem}\label{prop:backward_control}
Assume (A\ref{hyp:A}). Then there exist a $C<+\infty$ such that for any $k\in\{0,\dots,n-2\}$
$\epsilon>0$, $y_{0:n}$ and $\varphi\in\mathcal{B}_b(\mathbb{R}^{d_x})$ we have
$$
\sup_{x\in\mathbb{R}^{d_x}}|\int [M_{k+1,\hat{\eta}_{k,\epsilon},\epsilon}(x,dz) - M_{k+1,\hat{\eta}_k}(x,dz)]\varphi(z)| \leq C \epsilon
$$
where $M_{k+1,\hat{\eta}_{k,\epsilon},\epsilon}$ and $M_{k+1,\hat{\eta}_k}$ are defined in \eqref{eq:backward_true}-\eqref{eq:backward_abc} and $\hat{\eta}_{k,\epsilon}$ and $\hat{\eta}_k$ are the ABC and true filters.
\end{lem}

\begin{proof}
We have the decomposition:
$$
\int [M_{k+1,\hat{\eta}_{k,\epsilon},\epsilon}(x,dz) - M_{k+1,\hat{\eta}_k}(x,dz)]\varphi(z)
= 
\int \varphi(z)f(z,x)\bigg[\frac{\hat{\eta}_{k,\epsilon}(z) - \hat{\eta}_k(z)}{\int \hat{\eta}_{k,\epsilon}(u)f(u,x) du} 
$$
$$
+ \hat{\eta}_k(z)\bigg\{\frac{\int[\hat{\eta}_k(u) - \hat{\eta}_{k,\epsilon}(u)]f(u,x)du}
{\int \hat{\eta}_{k,\epsilon}(u)f(x,u)du\int \hat{\eta}_k(u)f(u,x) du}\bigg\}\bigg] dz
$$
where we have suppressed the data from the notation.

Dealing with the first part, we have for some $C$ that does not depend upon $x,k,\epsilon$ or $y_{0:n}$
$$
\int \varphi(z)f(z,x)\bigg[\frac{\hat{\eta}_{k,\epsilon}(z) - \hat{\eta}_k(z)}{\int \hat{\eta}_{k,\epsilon}(u)f(u,x) du}\bigg]dz \leq  C \epsilon \|\varphi\|
$$
where we have used (A\ref{hyp:A}) (\ref{hyp:density_control}) in the denominator and to control $f$ as well as Theorem 2 of \cite{jasra}.
Now, for the second part
$$
\int \varphi(z)f(z,x)
\hat{\eta}_k(z)\bigg\{\frac{\int[\hat{\eta}_k(u) - \hat{\eta}_{k,\epsilon}(u)]f(u,x)du}
{\int \hat{\eta}_{k,\epsilon}(u)f(x,u)du\int \hat{\eta}_k(u)f(u,x) du}\bigg\}\bigg] dz
 \leq \|\varphi\| C \epsilon 
$$
where, again (A\ref{hyp:A}) (\ref{hyp:density_control}) has been applied along with Theorem 2 of \cite{jasra} and $C$ does not depend upon $x,k,\epsilon$ or $y_{0:n}$. Using the uniformity in $x$ of the above bounds allows us to conclude.
\end{proof}

\end{document}